\documentclass[draftcls, peerreview]{IEEEtran}
\usepackage{amsmath,amsthm,amssymb,amsfonts,amscd,mathrsfs}

\usepackage{placeins}
\usepackage[dvips]{graphicx}
\usepackage[mathcal]{euscript}
\theoremstyle{plain}
\newtheorem{theorem}{Theorem}[section]
\newtheorem{lemma}[theorem]{Lemma}

\newtheorem{proposition}[theorem]{Proposition}

\theoremstyle{definition}
\newtheorem{definition}[theorem]{Definition}

\newtheorem{example}[theorem]{Example}

\theoremstyle{remark}
\newtheorem{remark}[theorem]{Remark}

\renewcommand{\bf}{\textbf}
\newcommand{\lp}{\langle}
\newcommand{\rp}{\rangle}

\newcommand{\Z}{\mathbb{Z}}

\newcommand{\cP}{\mathscr{P}}

\newcommand{\tC}{\widetilde{C}}

\newcommand{\hC}{\widehat{C}}

\newcommand{\hf}{\widehat{f}}

\newcommand{\tf}{\tilde f}

\newcommand\vlambda{{\boldsymbol{\lambda}}}

\newcommand{\vc}{\mathbf{c}}
\newcommand{\vg}{\mathbf{g}}

\newcommand{\vv}{\mathbf{v}}
\newcommand{\vw}{\mathbf{w}}

\newcommand{\vb}{\mathbf{b}}

\newcommand{\ve}{\mathbf{e}}
\newcommand{\vx}{\mathbf{x}}
\newcommand{\vy}{\mathbf{y}}
\newcommand{\vz}{\mathbf{z}}

\newcommand{\vzero}{\mathbf{0}}

\newcommand{\F}{\mathbb{F}}

\newcommand{\cV}{\mathcal{V}}

\newcommand{\Gal}{\textrm{Gal}}
\newcommand{\GL}{\operatorname{GL}}

 \newcommand{\GammaL}{\Gamma
\operatorname{L}}  
\newcommand{\Mat}{\textrm{Mat}}
\newcommand{\RM}{\textrm{RM}}
\newcommand{\Vector}{\textrm{Vec}}
\newcommand{\epsB}{\epsilon_{\vb}}
\newcommand{\epsg}{\epsilon_{\vg}}

\newcommand{\Matlm}{\F_q^{l \times m}}

\newcommand{\Matml}{\F_q^{m \times l}}
\newcommand{\Span}{\operatorname{span}} 
\newcommand{\rank}{\operatorname{rank}}
\newcommand{\LEquiv}{LEquiv}
\newcommand{\SLEquiv}{SLEquiv}
\newcommand{\LAut}{LAut}
\newcommand{\SLAut}{SLAut}

\title{Equivalence for Rank-metric and Matrix Codes\\ and Automorphism Groups of Gabidulin Codes}
\author{Katherine Morrison
\thanks{This work was supported in part by NSF grants grants DMS-0735099, DMS-0903517 and DMS-0838463, as well as the DOE grant P200A090002.}
\thanks{K.\ Morrison is with the School of Mathematical Sciences, University of Northern Colorado, Greeley, CO 80639, USA 
  {\tt\small katherine.morrison@unco.edu}}%
 \thanks{ This work stemmed from her Ph.D. thesis, which was completed under the supervision of Judy L.\ Walker in the Department
  of Mathematics, University of Nebraska, Lincoln, NE 68588-0130, USA}%

}

\begin{document}
\maketitle
\begin{abstract}
For a growing number of applications such as cellular, peer-to-peer, and sensor networks, efficient error-free transmission of data through a network is essential.  Toward this end, K\"{o}tter and Kschischang propose the use of subspace codes to provide error correction in the network coding context.  The primary construction for subspace codes is the lifting of rank-metric or matrix codes, a process that preserves the structural and distance properties of the underlying code.  Thus, to characterize the structure and error-correcting capability of these subspace codes, it is valuable to perform such a characterization of the underlying rank-metric and matrix codes.  This paper lays a foundation for this analysis through a framework for classifying rank-metric and matrix codes based on their structure and distance properties.

To enable this classification, we extend work by Berger on equivalence for rank-metric codes to define a notion of equivalence for matrix codes, and we characterize the group structure of the collection of maps that preserve such equivalence.  We then compare the notions of equivalence for these two related types of codes and show that matrix equivalence is strictly more general than rank-metric equivalence.  Finally, we characterize the set of equivalence maps that fix the prominent class of rank-metric codes known as Gabidulin codes. In particular, we give a complete characterization of the rank-metric automorphism group of Gabidulin codes, correcting work by Berger, and give a partial characterization of the matrix-automorphism group of the expanded matrix codes that arise from Gabidulin codes.
\end{abstract}

\section{Introduction}\label{intro}
To enable efficient transmission of information through a network, Ahlswede, et al.\ \cite{ACLY00} propose a method known as \emph{network coding}.  With this approach, each node has the ability to
intelligently combine, or \emph{code}, the information coming into it, and then pass along this new encoded data toward the sink; this coding often enables the network to achieve a higher \emph{throughput}, i.e.\ a higher amount of information transmitted from the source to each receiver in a single unit of time.  For certain classes of networks, K\"{o}tter and M\'edard \cite{KM03} prove that it is sufficient to use \emph{random linear network coding} to achieve capacity; in other words, it is sufficient to simply allow each internal node to pass along a randomly generated linear combination of its inputs as long as the coefficients of each linear combination live over a sufficiently large finite field, and this method produces the largest possible throughput for the network.  In addition to achieving capacity, random linear network coding has the added benefit that code design and implementation are independent of the network topology, and so can be used in cases where the network topology is unknown or often changing, as with peer-to-peer networks. As a result, random linear network coding is highly appealing for widespread implementation.

A significant drawback of network coding arises, however, when noise is introduced at any of the internal nodes or links.  Even a single error introduced somewhere early in the network can propagate through to potentially corrupt all the final outputs; thus, some form of error correction is necessary.  Since random linear network coding outputs linear combinations of the input vectors, the subspace of input vectors is preserved at the output.  K\"{o}tter and Kschischang \cite{KK08} propose the use of \emph{subspace codes}, i.e.\ carefully chosen collections of subspaces, to provide error correction for random linear network coding.  They also propose a simple construction for subspace codes via the lifting of linear codes whose codewords are either matrices over $\F_q$ or whose codewords are vectors over $\F_{q^m}$ equipped with the rank distance.  Codes with matrices as codewords are also known as \emph{matrix codes}, array codes e.g.\ \cite{Blaum}, or space-time codes over a finite field e.g.\ \cite{GV08}, while the codes consisting of vectors over $\F_{q^m}$ are known as \emph{rank-metric codes}. In particular, K\"{o}tter and Kschischang propose lifting Gabidulin codes, which comprise a class of optimal rank-metric codes. Additionally, they introduce a metric on the collection of subspaces and define a minimum-distance decoder for subspace codes.  The subspace metric turns out to be a scalar multiple of the rank metric when the subspaces are lifted from matrix or rank-metric codes \cite{SKK08}, and so it is valuable to study the structure and distance properties of such codes.  Thus, the primary focus of this work is to provide a framework for classifying rank-metric and matrix codes based on these defining properties.  

 In Section \ref{background}, we give some necessary background on subspace codes with a focus on the lifted matrix code and lifted rank-metric code constructions.  We review K\"{o}tter and Kschischang's foundational subspace code construction of lifted Gabidulin codes, and so we give a definition of Gabidulin codes as well.  We then turn to an analysis of the underlying rank-metric and matrix codes.  To enable this analysis, we begin by characterizing the relationship between these two classes of codes via the linear map $\epsB$, which expands elements of an $\F_q$-extension field with respect to an ordered basis $\vb$ for that field as an $\F_q$-vector space.  

In Section \ref{equivalence}, we then turn toward classifying rank-metric and matrix codes in terms of their structural and distance properties.  This leads us to a definition of \emph{equivalence} for each class of codes.  With this in place, in Subsections \ref{RMEquivMaps} and \ref{MatEquivMaps}, we characterize the collections of linear and semi-linear equivalence maps for rank-metric codes, correcting a result from \cite{Berger03}, and contrast these with the collections of linear and semi-linear equivalence maps for matrix codes, again appealing to the map $\epsB$ to navigate between these two code types.

In Section \ref{autgroupssection}, we investigate the subset of linear equivalence maps that fix a given code, which is termed the \emph{linear automorphism group} of the code.    We provide a complete characterization of the linear automorphism group of the class of rank-metric codes known as Gabidulin codes in Subsection \ref{RMAutGroup}.  Berger previously attempted to characterize this group, but we have found a flaw in his proof and provide counterexamples to his characterization.  Finally, in Subsection \ref{MatAutGroup}, we give a partial characterization of the linear automorphism group of the matrix codes that arise from expanding Gabidulin codes via the map $\epsB$.  

\section{Background}\label{background}
\subsection{Subspace Codes for Random Linear Network Coding}\label{subspacecodes}
As outlined in the Section \ref{intro}, since random linear network coding only preserves the subspace of information vectors injected by the source, and errors introduced by the network will shift this subspace to another ``nearby" subspace, K\"{o}tter and Kschischang propose the use of subspace codes with an appropriate distance metric to provide error correction in this context \cite{KK08}.  This motivates the following definitions.

\begin{definition}[\cite{KK08}]\label{subspacedist}
 Let $\cV$ be an $n$-dimensional vector space over
  $\F_q$ and let $\cP(\cV)$ be the collection of subspaces of $\cV$.
  A \emph{subspace code} is a non-empty subset of $\cP(V)$.  For any
  $U, V \in \cP(\cV)$, the \emph{subspace distance} between $U$ and
  $V$ is defined as $d_S(U,V):=\dim (U+V) - \dim( U\cap V).$
\end{definition}

For ease of decoding, K\"{o}tter and Kschischang restrict to codes where all the subspace codewords have the same dimension; thus, each receiver knows to stop collecting additional vectors from the network once it has accumulated a set number of linearly independent output vectors.  These subspace codes are known as \emph{constant dimension} codes:
\begin{definition} A subspace code $C \subseteq \cP(\cV)$ is called a
  \emph{constant-dimension subspace code} if $\dim U=\dim V$ for all $U, V \in C$. 
  \end{definition}

A constant-dimension subspace code $C \subseteq \cP(\cV)$ over $\F_q$ is typically
described by a 4-tuple of parameters $[n,l,k, d_{S,\min}]_q$, where $n$ is
the dimension of $\cV$; $l$ is the dimension of each $U \in C$, which is necessarily less than or equal to $n$; $k =
\log_q|C|$ is the log of the number of codewords in $C$; and $d_{S,\min}
= d_{S,\min}(C)$ is the minimum subspace distance between any two
distinct codewords in $C$.

\subsection{Lifted Matrix Codes}
We may fix a basis for the $n$-dimensional ambient $\F_q$-vector space $\cV$ and identify $\cV$ with $\F_q^n$ via this
choice of basis.  There is then a one-to-one correspondence between
$l$-dimensional subspaces of $\cV$ and $l \times n$ matrices in
reduced row echelon form, given by $U \leftrightarrow M$, where $M$ is
the unique matrix in reduced row echelon form whose rows form a basis
for $U$.  Recall that the \emph{pivot locations} of an $l \times n$
matrix $M$ in reduced row echelon form are the integers $i$, $1 \leq i
\leq n$, such that the $i^{\text{th}}$ column of $M$ is a standard
basis vector.  If $M$ has rank $l$, then there will be precisely $l$
pivot locations; if the pivot locations are $i_1<i_2<\dots<i_l$, then
the $l \times l$ matrix whose $s^{\text{th}}$ column is the
$i_s^{\text{th}}$ column of $M$ is simply the identity matrix.  If $M$ is the
matrix in reduced row echelon form corresponding to the subspace $U$
of $\cV$, then we will abuse terminology and refer to the pivot
locations of $M$ as the pivot locations of $U$.

\begin{definition}[\cite{SKK08}]\label{liftedmatrixcode}
Let $C$ be an $[n,l,k,d_{S,\min}]$ constant-dimension
  subspace code and suppose the 
  pivot locations for every $U \in C$ coincide.  Then we call $C$ a
  \emph{lifted matrix code}. For $U \in C$, let $M$ be the corresponding matrix in
  reduced row echelon form and let $1 \leq j_1<j_2<\dots <j_{n-l}\leq
  n$ be the non-pivot locations of $M$.  The $l \times (n-l)$ matrix
  $A$ whose $s^{\text{th}}$ column is the $j_s^{\text{th}}$ column of
  $M$ is the \emph{auxiliary matrix} for $U$.  We call the collection   
  \[
  \hC =  \{A \in \F_q^{l \times (n-l)}~|~ A \text{ is an auxiliary matrix of some } U \in C\}
  \]
of auxiliary matrices of codewords of $C$ the \emph{underlying matrix code of $C$}, and we define the \emph{minimum rank-distance} $d_{R,min}$ of $\hC$ to be the minimum over all $A, B \in \hC$ of the rank distance $d_R(A,B):=\rank(A-B)$.
\end{definition}

The previous definition outlines how to identify that a subspace code is a lifted matrix code, but one may also easily construct a subspace code by lifting a matrix code of appropriate dimensions.  In particular, once $l$ columns are selected as pivot locations and a matrix code $\hC \subseteq \F_q^{l \times (n-l)}$ is chosen, the lifted matrix code is obtained by appropriately interspersing the columns of the matrix codewords with the pivot columns.  Observe that each choice of pivot columns will result in a distinct constant-dimension subspace code; however, the distance distribution is independent of the choice of pivot columns, and is determined only by that of the underlying matrix code.  This result is captured in Lemma \ref{subspacedistance} from \cite{SKK08}.

\begin{lemma}[Silva, Kschischang, K\"{o}tter \cite{SKK08}] \label{subspacedistance}Suppose $U,V
  \in \cP(V)$ are $l$-dimensional subspaces of $\cV$ that have the
  same pivot locations and let $A$ and $B$ be the $l \times (n-l)$
  auxiliary matrices for $U$ and $V$, respectively.  Then $d_S(U,V) = 2\rank(A-B) = 2d_R(A,B)$.
Hence, if $C$ is a lifted matrix code and $\hC$ is the corresponding matrix code, then $d_{S,min}(C) = 2d_{R,min}(\hC).$
\end{lemma}

Given this simple correspondence between the distance distributions of lifted matrix codes and their underlying matrix codes, a natural next direction is to seek out constructions of matrix codes with good distance properties.  Toward this end, we examine the relationship between matrix codes and another class of codes known as \emph{rank-metric} codes.  To facilitate this, we first need to fix some notation.

\subsection{Lifted Rank-Metric Codes}
\begin{definition}
Fix an ordered basis $\vb = (b_1, \ldots, b_m)$ for $\F_{q^m}$ over $\F_q$.  The \emph{vector expansion with respect to $\vb$} is the map $\epsB: \F_{q^m}\to \F_q^m$ given by $\epsB(a)=(a_1,a_2, \ldots, a_m)$ where $a=a_1b_1+a_2b_2+\ldots +a_{m}b_m$.  
\end{definition}
\begin{remark}\label{epsBlinear}
Observe that $\epsB$ is $\F_q$-linear, but not a field homomorphism.
\end{remark}

In addition to expanding elements of $\F_{q^m}$ to form vectors in $\F_q^m$, we will also need to consider expanding vectors in $\F_{q^m}^l$ to form matrices in $\Matlm$.  As a slight abuse of notation, we will denote the map for this expansion by $\epsB$ as well.  Thus, we obtain the following definition.

\begin{definition}
Fix an ordered basis $\vb=(b_1, \ldots, b_m)$ for $\F_{q^m}$ over $\F_q$.  \emph{Matrix expansion with respect to $\vb$} is the map $\epsB:\F_{q^m}^l~\to~\Matlm$ given by 
$$\epsB(\vx)={\footnotesize  \begin{bmatrix}\epsB(x_1)\\ \epsB(x_2)\\\vdots\\ \epsB(x_{l})\end{bmatrix}}$$ 
where $\vx=(x_1, x_2, \ldots, x_l)$.  
\end{definition}

With this notation in place, we now define the notion of rank-metric distance and thus rank-metric codes.
\begin{definition}[\cite{Gabidulin85}]\label{RankMetricDef}
Fix an ordered basis $\vb$ for $\F_{q^m}$ over $\F_q$.  For $\vx, \vy \in \F_{q^m}^l$, the \emph{rank-metric distance} between
$\vx$ and $\vy$ is
\begin{eqnarray*}
d_{R}(\vx,\vy) &=& \text{dim}_{\F_q} \left(\text{span}_{\F_q} \{x_1 - y_1,~ x_2 -y_2,  \ldots,~ x_l-y_l \} \right)\\
&=& \rank(\epsB(\vx)-\epsB(\vy)).
\end{eqnarray*}
The \emph{rank-metric weight}, or \emph{rank}, of a vector $ \vx \in \F_{q^m}^l$ is 
\begin{eqnarray*}
d_{R}(\vx,\vzero) &=& \text{dim}_{\F_q} \left(\text{span}_{\F_q} \{x_1 ,~ x_2,  \ldots,~ x_l\} \right)=\rank\epsB(\vx).
\end{eqnarray*}
A \emph{rank-metric code} of length $l$ and minimum rank-metric distance
$d_{R,\min} = d_{R,\min}(C)$ over $\F_{q^m}$ is a subset $C$ of
$\F_{q^m}^l$ such that $d_{R, \min} = \displaystyle \hspace{-.35in}\min_{{\scriptsize \begin{array}{c} \vx, \vy \in C,~  \vx\neq \vy \end{array}}} \hspace{-.3in}d_R(\vx, \vy).$\\
\vspace{-.2in}\\
If $C \subseteq
\F_{q^m}^l$ is a rank-metric code then
\[
\epsB(C)= \{\epsB(\vx) \, | \, \vx \in C\} \subseteq \Matlm
\]
is called the \emph{expanded matrix code} of $C$.
\end{definition}
\begin{remark}
Since the rank-metric distance between two vectors equals the rank distance between their corresponding matrix expansions, the rank-metric distance on vectors is equivalent to the rank distance on matrices.  Furthermore, the rank-metric distance is independent of the choice of basis $\vb$ for the ambient space.  For this reason, we denote both distances by $d_R$, and we assume that the context will make it clear whether it is necessary to first apply $\epsB$ to evaluate that distance measure.  
\end{remark}
Since it is possible to obtain a matrix code from any rank-metric code, we also have a notion of lifted rank-metric codes.
\begin{definition}
Fix an ordered basis $\vb$ for $\F_{q^m}$ over $\F_q$.  Let $C \subseteq \F_{q^{m}}^l$ be a rank-metric code with expanded matrix code $\epsB(C)$.  Given $l$ integers, $1\leq i_1<i_2<\dots<i_l\leq l+m$, the \emph{lifted rank-metric code} of $C$ is the lifted matrix code $\hC \subseteq \cP(\F_q^{l+m})$, as in Definition \ref{liftedmatrixcode}, whose pivot columns are $ i_1, i_2, \ldots, i_l$ and whose underlying matrix code is $\epsB(C)$.
\end{definition}
By Lemma \ref{subspacedistance} and the definition of rank-metric distance for rank-metric codes, we see that if $C$ is a rank-metric code of minimum rank-metric distance $d_{R,min}$, then the lifted rank-metric code $\hC$ has minimum subspace distance $d_{S,min}=2d_{R,min}$.  Thus, any $[l,k,d_{R,min}]_{q^m}$ rank-metric code gives rise to an $[l+m, l, km, 2d_{R,min}]_q$ lifted rank-metric code. 

K\"{o}tter and Kschischang first proposed the construction of lifting rank-metric codes in their seminal paper \cite{KK08}, where they focused specifically on lifting the family of rank-metric codes that have become known as \emph{Gabidulin codes}.  These codes are $q^m$-ary analogues of Reed-Solomon codes that are optimal for the rank-metric distance in that they meet a rank-metric analogue of the Singleton bound.  We review two constructions of Gabidulin codes below as these are the most prominent rank-metric codes.  

In keeping with Gabidulin's original notation, we will use $a^{[i]}$ to mean $a^{q^i}$ for any $a \in \F_{q^m}$ and integer $i$.
\begin{definition}[Theorems 6 and 7 in  \cite{Gabidulin85}]\label{GabCodeDef}
An $[n,k, d]_{q^m}$ \emph{Gabidulin code} $C$, with $k=n-d+1$ and $n<m$, is a code defined by a parity-check matrix of the form
\[
H={\footnotesize \begin{bmatrix} h_1 &h_2 &\ldots &h_n \\ h_1^{[1]} & h_2^{[1]} & \ldots  &h_n^{[1]}\\ \vdots &\vdots &\vdots& \vdots \\ h_1^{[d-2]} & h_2^{[d-2]}  &\ldots  &h_n^{[d-2]} \end{bmatrix},}
\]
where $\{h_i \in \F_{q^m}~|~1 \leq i \leq n\}$ are linearly independent over $\F_q$.  
Equivalently, an $[n,k, d]_{q^m}$ \emph{Gabidulin code} $C$, with $k=n-d+1$ and $n<m$, is a code defined by a generator matrix of the form
\[
G={\footnotesize \begin{bmatrix} g_1  &g_2 & \ldots & g_n \\ g_1^{[1]}  &g_2^{[1]}  &\ldots  &g_n^{[1]}\\ \vdots& \vdots &\vdots &\vdots \\ g_1^{[k-1]} & g_2^{[k-1]} & \ldots  &g_n^{[k-1]} \end{bmatrix},}
\]
where $\{g_i \in \F_{q^m}~|~1 \leq i \leq n\}$ are linearly independent over $\F_q$.  For compactness, we denote such a Gabidulin code by $C_{k,\vg, q^m}$ where $\vg=(g_1, g_2, \ldots, g_n)$ similar to the notation in \cite{Berger03}.  Any vector $\vg \in \F_{q^m}^n$ whose entries are  linearly independent over $\F_q$ will be called a \emph{Gabidulin vector}, since such a vector can be used to define a Gabidulin code.  
\end{definition}

In analogy with Reed-Solomon codes, Gabidulin codes may equivalently be defined in terms of evaluating linearized polynomials at a collection of linearly independent points in $\F_{q^m}$  \cite{Gabidulin85}.  This equivalent definition allows for the creation of efficient encoding and decoding algorithms.  We will not need that construction here, however, and so we refer the interested reader to \cite{Gabidulin85} for further details.

In \cite{KK08}, K\"{o}tter and Kschischang prove a Singleton bound for constant-dimension subspace codes in analogy with the Singleton bound for block codes.  They then give an asymptotic version of this bound and show that the family of lifted Gabidulin codes aymptotically achieves this bound.  Given the asymptotic optimality of certain lifted rank-metric and matrix codes, we are led to further investigate the structure and distance properties of underlying rank-metric and matrix codes as these may prove valuable for further lifted subspace code constructions.  To enable this investigation, we must be able to classify codes based on these structural and distance properties, and so we are led here to define and examine an appropriate notion of code \emph{equivalence}.  In the following section, we will examine the appropriate notion of equivalence for rank-metric codes as well as the notion of equivalence for matrix codes, with an eye towards comparing and contrasting the resulting equivalence maps.  

To enable this comparison of equivalence maps, we need a method for translating between rank-metric and matrix codes since each equivalence map is only defined to operate on one of these code types.  Thus far, we have seen that to any rank-metric code $C \subseteq \F_{q^m}^l$, we may associate a matrix code $\epsB(C)\subseteq \Matlm$ by expanding $C$ with respect to some ordered basis $\vb$ for $\F_{q^m}$ as an $\F_q$-vector space.  But to enable our comparison, we must also have a map to translate back from matrix codes to rank-metric codes; this is accomplished via $\epsB^{-1}$, which compresses the matrix code $\hC$ with respect to the basis $\vb$. This notion of compression and the mechanism for accomplishing it is made more precise in the following definition.

\begin{definition}\label{compression}
Fix an ordered basis $\vb=(b_1,b_2, \ldots, b_m)$ for $\F_{q^m}$ over $\F_q$.  \emph{Matrix compression with respect to $\vb$} is the map $\epsB^{-1}: \Matlm \to \F_{q^m}^l$ given by 
\begin{eqnarray*}
\epsB^{-1}(X)&=&\left(\sum_{j=1}^m x_{1 j}b_j, \sum_{j=1}^mx_{2 j}b_j, \ldots, \sum_{j=1}^m x_{l j}b_j\right)\\
& =& \left(X(b_1, \ldots, b_m)^\top\right)^\top\\
&=& (b_1, \ldots, b_m) X^\top
\end{eqnarray*}
 where $X=[x_{ij}]\in \Matlm$.    
If $\hC \subseteq \Matlm$ is a matrix code, then $\epsB^{-1}(\hC)= \{\epsB^{-1}(X) \, | \, X \in \hC\} \subseteq \F_{q^m}^l$
is called the \emph{compressed rank-metric code of $\hC$}.
\end{definition}

\section{Equivalence for Rank-metric and Matrix Codes} \label{equivalence}
Intuitively, two codes should be considered equivalent if they share all the same properties and structure.  In particular, equivalent codes should have the same distance distribution and the same number of codewords, or dimension if the codes are linear.  To preserve the dimension of a linear code, any map between equivalent codes must take a subspace to a subspace of the same dimension; we will term such a map to be \emph{subspace-preserving}.  While we would additionally desire that an equivalence map be distance-preserving, this characteristic is significantly harder to enforce than simply requiring that the map be weight-preserving.  It is well-known, however, that for linear codes the distance and weight distributions coincide, and so any additive map between linear codes that is subspace-preserving and weight-preserving is also distance-preserving.  Thus, we simplify the notion of equivalence maps as follows:
 we say a map between codes is an \emph{equivalence map} if it is additive, weight-preserving, and subspace-preserving.  While there is a broad collection of subspace-preserving maps, we will restrict to the class of \emph{semi-linear} maps (a notion that will be made precise below) because those have the greatest structure in terms of linearity that still allows for possible renaming/reordering of elements of $\F_q$.  We will also consider restriction to the class of \emph{linear} maps because analysis of these maps is generally simpler and cleaner than that of semi-linear maps, and is often a necessary first step to characterizing the semi-linear maps.  

 \begin{definition}[\cite{HuffmanChapter}]
A map $f: \F_q^n \to \F_q^n$ is \emph{semi-linear} if it satisfies the following conditions:
\begin{enumerate}
\item $f(\vx+ \vy) = f(\vx)+f(\vy)$ for every $\vx, \vy \in \F_q^n$, and 
\item there is some $\gamma \in \Gal(\F_q/\F_p)$ such that $f(\alpha\vx)=\alpha^\gamma f(\vx)$ for every $\alpha \in \F_q$ and every $\vx \in \F_q^n$, where $q=p^e$ for some $e$.
\end{enumerate}
\noindent The collection of invertible semi-linear maps on $\F_q^n$ under composition is called the \emph{general semi-linear group} and is denoted $\GammaL_n(\F_q)$, in analogy with the general linear group $\GL_n(\F_q)$.
\end{definition}

It is clear that any invertible semi-linear map is subspace-preserving since it is simply the composition of a linear map and an $\F_p$-automorphism of $\F_q$.  In fact, the collection of invertible semi-linear maps forms a subgroup of the group of subspace-preserving maps. More precisely, $\GammaL_n(\F_q)$ is the semi-direct product of the collection of invertible linear maps with the collection of $\F_p$-automorphisms of $\F_q$ \cite{HuffmanChapter}.  A typical element of $\GammaL_n(\F_q)= \GL_n(\F_q)~\rtimes~\Gal(\F_q/\F_p)$ is an ordered pair $(A; \gamma)$, which acts on $\F_q^n$ via $\vx(A; \gamma) = (\vx A)^\gamma$, where $\gamma$ acts on $\vx A\in \F_q^n$ coordinate-wise.  Note that we may also consider $\gamma \in \Gal(\F_q/\F_p)$ as acting on $\Matlm$ coordinate-wise, which enables us to express multiplication in the group $\GammaL_n(\F_q)$ by 
\[
(A_1; \gamma_1)(A_2; \gamma_2) = \left(A_1A_2^{\left(\gamma_1^{-1}\right)}; \gamma_1\gamma_2\right)
\]
since we have 
\[
\vx(A_1; \gamma_1)(A_2; \gamma_2) = \left((\vx A_1)^{\gamma_1}\right)(A_2; \gamma_2) = \left((\vx A_1)^{\gamma_1}A_2\right)^{\gamma_2} = \left(\vx A_1 A_2^{\gamma_1^{-1}}\right)^{\gamma_1\gamma_2}.
\]

\subsection{Review of Block Code Equivalence Maps}
With the background on semi-linear maps in place, we may now define precisely the notion of equivalence for rank-metric codes in $\F_{q^m}^l$ and for matrix codes in $\Matlm$.  To place this work in context, we begin by quickly recalling the notion of equivalence for block codes in $\F_q^n$.
 \begin{definition}[\cite{HuffmanChapter}]\label{blockequivalence}
 An invertible map $f:\F_q^n \to \F_q^n$ is a \emph{linear block-equivalence map} if $f$ is $\F_q$-linear and preserves Hamming weight.  Similarly, an invertible map $f:\F_q^n \to \F_q^n$ is a \emph{semi-linear block-equivalence map} if $f$ is $\F_q$-semi-linear and preserves Hamming weight. Two block codes $C, \tC \subseteq \F_q^{n}$ are \emph{(semi-)linearly block-equivalent} if there exists a (semi-)linear block-equivalence map $f$ such that $\tC= f(C)$.
  \end{definition}
    \begin{remark}
Since the composition of two linear Hamming weight-preserving maps is also a linear Hamming weight-preserving map, the collection of linear block equivalence maps forms a group under composition.  An analogous statement holds in the semi-linear case. 
  \end{remark}
Recall that a monomial matrix is a matrix that has precisely one non-zero entry in each row and each column.  Any monomial matrix can be written in the form $DP$ where $D$ is an invertible diagonal matrix and $P$ is a permutation matrix 
\cite{SelfDualChapter}.  As a consequence of the MacWilliams Extension Theorem \cite{MacWilliamsThesis}, the monomial matrices are the only Hamming-weight preserving linear maps.  Thus, the collection of linear block-equivalence maps is precisely the subgroup of monomial matrices \cite{HuffmanChapter}.  Furthermore, since field automorphisms preserve Hamming weight, the collection of semi-linear block-equivalence maps is the subgroup formed from the semi-direct product of the monomial matrices with the group $\Gal(\F_q/\F_p)$.
%
%
%
 \subsection{Rank-Metric Code Equivalence Maps}\label{RMEquivMaps}
  
 We now turn to the notion of equivalence for rank-metric codes, which was first studied by Berger in  \cite{Berger03}. 
 \begin{definition}[\cite{Berger03}]\label{RMequivalence}
 An invertible map $f:\F_{q^m}^l \to \F_{q^m}^l$ is a \emph{linear rank-metric-equivalence map} if $f$ is $\F_{q^m}$-linear and preserves rank weight.   The collection of linear rank-metric-equivalence maps is denoted by $\LEquiv_{\RM}(\F_{q^m}^l)$. Similarly, an invertible map $f:\F_{q^m}^l \to \F_{q^m}^l$ is a \emph{semi-linear rank-metric-equivalence map} if $f$ is $\F_{q^m}/\F_p$-semi-linear and preserves rank weight.   The collection of semi-linear rank-metric-equivalence maps is denoted by $\SLEquiv_{\RM}(\F_{q^m}^l)$. Two rank-metric codes $C, \tC \subseteq \F_{q^m}^l$ are \emph{(semi-)linearly rank-metric-equivalent} if there exists a (semi-)linear rank-metric-equivalence map $f$ such that $\tC= f(C)$.
  \end{definition}
   \begin{remark}
Since the composition of two linear rank weight-preserving maps is also a linear rank weight-preserving map, the collection of linear rank-metric equivalence maps forms a group under composition.  An analogous statement holds in the semi-linear case. 
  \end{remark}
\begin{remark} 
Berger  \cite{Berger03} refers to linear rank-metric equivalence maps as \emph{linear isometries} and uses the notation $Iso(\F_{q^m}^l)$ to refer to the collection of such maps.  Similarly, he terms semi-linear rank-metric equivalence maps as \emph{semi-linear isometries} and uses the notation $SIso(\F_{q^m}^l)$ to refer to the collection of such maps.  We will use the term \emph{equivalence map} here, however, for consistency with the block code literature examining equivalence classes of codes. 
\end{remark}

In  \cite{Berger03}, Berger precisely describes the collections of linear and semi-linear rank-metric-equivalence maps.  He shows that the only linear rank-metric-equivalence maps are multiplication by non-zero scalars in $\F_{q^m}$, multiplication on the right by elements of $\GL_{l}(\F_q)$, and compositions thereof; he also shows that the only semi-linear rank-metric equivalence maps are given by a linear rank-metric equivalence map composed with coordinate-wise application of automorphisms in $\Gal(\F_{q^m}/\F_p)$. Propositions \ref{LEquivRM} and \ref{SLEquivRM} give the group structure for these two collections of maps.

\begin{proposition}\label{LEquivRM}
 The group of linear rank-metric-equivalence maps on $\F_{q^m}^l$ satisfies
 \[
\LEquiv_{\RM}(\F_{q^m}^l)\cong \left(\F_{q^m}^* \times \GL_l(\F_q)\right) /N,
\] 
where $N=\{(\lambda, \lambda^{-1}I_l)~|~\lambda~\in~\F_q^*\} \leq \F_{q^m}^* \times \GL_l(\F_q)$. 
\end{proposition}
\begin{proof}
In \cite{Berger03}, Berger proves that the linear rank-metric equivalence maps can be represented as products of $\F_{q^m}^*$-scalar matrices and invertible matrices over $\F_q$, and so 
\[
\LEquiv_{\RM}(\F_{q^m}^l) = \{\alpha I_l~|~ \alpha \in \F_{q^m}^*\}\cdot \GL_l(\F_q) \subseteq \GL_l(\F_{q^m}).
\]
 Since the scalar matrices form a normal subgroup, the product of that subgroup with $\GL_l(\F_q)$ is well-defined and does in fact form a group.  Thus, $$|\LEquiv_{\RM}(\F_{q^m}^l)| = \frac{| \{\alpha I_l~|~ \alpha \in \F_{q^m}^*\}| |\GL_l(\F_q)|}{ |  \{\alpha I_l~|~ \alpha \in \F_{q^m}^*\}\cap\GL_l(\F_q)|}= \frac{(q^m-1)\prod_{i=0}^{l-1}(q^l-q^i)}{q-1}.$$  Furthermore, since the scalar matrices are in the center of $\GL_l(\F_{q^m})$, it is possible to represent each equivalence map as a single scalar multiplication followed by multiplication by a single matrix in $\GL_l(\F_{q})$.  Thus, there is a natural homomorphism from the direct product $\F_{q^m}^* \times \GL_l(\F_q)$ onto $\LEquiv_{\RM}(\F_{q^m}^l)$.  Since it is possible to represent multiplication by an $\F_q^*$-scalar in two ways, namely via an $\F_{q^m}^*$-scalar or by an $\F_q^*$-scalar matrix, we see that the kernel of this map contains the subgroup $N=\{(\lambda, \lambda^{-1}I_l)~|~ \lambda \in \F_q^*\}$.  Finally, since $|N|=q-1$, we have 
 \[
 |\left(\F_{q^m}^* \times \GL_l(\F_q)\right) /N| =  \frac{(q^m-1)\prod_{i=0}^{l-1}(q^l-q^i)}{q-1} = |\LEquiv_{\RM}(\F_{q^m}^l)|, 
 \]
  and so $N$ must equal the kernel of the map.  Thus, by the first isomorphism theorem, the result holds. 
\end{proof}
\begin{remark}
By the Proposition \ref{LEquivRM}, each linear rank-metric equivalence map corresponds to a coset of the form $(\alpha, L) \cdot N$ for some $\alpha \in \F_{q^m}^*$ and $L \in \GL_l(\F_q)$.  For ease of notation, we will henceforth write $[\alpha, L]$ to denote the coset $(\alpha, L) \cdot N$.  
\end{remark}

In Proposition \ref{LEquivRM}, we show that $\LEquiv_{\RM}(\F_{q^m}^l)\cong \left(\F_{q^m}^* \times \GL_l(\F_q)\right) /N$, while \cite{Berger03} previously asserted that $\LEquiv \cong (\F_{q^m}^* / \F_q^*) \times \GL_l(\F_q)$ by taking the direct product $\F_{q^m}^* \times \GL_l(\F_q)$ and modding out by the intersection of those groups.  However, this intersection does not give the maps that correspond to the identity map in $\LEquiv_{\RM}(\F_{q^m}^l)$, and so there is a flaw in the proof.  In particular, there are a number of values of $q$, $l$, and $m$ for which the two groups, $ \left(\F_{q^m}^* \times \GL_l(\F_q)\right)/N$ and $(\F_{q^m}^* / \F_q^*) \times \GL_l(\F_q)$, are not isomorphic.  Example \ref{BergerCounterexample} gives some insight into why the groups are not isomorphic in general.  

\begin{example}\label{BergerCounterexample}
Set $q=3$, $l=2$, and $m=4$.  Let $\alpha$ be a primitive element for $\F_{q^m}=\F_{81}$.  Consider the element $[\alpha, I_2] \in \left(\F_{81}^\cdot  \times \GL_2(\F_3)\right)/N$.  Since $([\alpha, I_2])^i = [\alpha^i, I_2]$ and the only element in $N$ with $I_2$ as its second coordinate is $(1, I_2)$, we see that $(\alpha^i,I_2) \in N$ precisely when $\alpha^i =1$.  Since $\alpha$ is a primitive element, its order is $q^m-1=80$, and so the order of $[\alpha, I_2]$ is also $80$.  

In contrast, we will show that the group $(\F_{81}^* / \F_3^*) \times \GL_2(\F_3)$ has no elements of order 80.  To see this, recall that the order of an ordered pair in a direct product equals the least common multiple of the orders of each entry of the ordered pair, and so $([\beta], B) \in (\F_{81}^* / \F_3^*) \times \GL_2(\F_3)$ has order equal to the least common multiple of the orders of $[\beta]=\beta \cdot \F_3^* \in (\F_{81}^* / \F_3^*) $ and $B \in \GL_2(\F_3)$.  The orders of elements in $(\F_{81}^* / \F_3^*)$ must divide the order of the group, which is $(81-1)/(3-1) = 40$, while the orders of elements in $ \GL_2(\F_3)$ must divide the order of the group, which is $\prod_{i=0}^{2-1}(3^2-3^i) = 48$.  Based on these order constraints, for $([\beta], B)$ to have order 80, we see that $[\beta]$ must have order 5 and $B$ must have order 16.  Using the computer algebra system Magma, we may check the order of each of the 48 elements of $\GL_2(\F_3)$, and we find that there is no element of order 16.  Thus, $(\F_{81}^* / \F_3^*) \times \GL_2(\F_3)$ has no element of order 80, and so it cannot be isomorphic to $\left(\F_{81}^* \times \GL_2(\F_3)\right) /N$ since isomorphic groups have the same number of elements of a given order.
%
\end{example}

More generally, we are interested in the collection of semi-linear rank-metric equivalence maps and their group structure.  These maps were previously investigated in \cite{Berger03}, and their structure was characterized:
\begin{proposition}[\cite{Berger03}]\label{SLEquivRM}
The group of semi-linear rank-metric-equivalence maps on $\F_{q^m}^l$ satisfies 
\[
\SLEquiv_{\RM}(\F_{q^m}^l)\cong \LEquiv_{\RM}(\F_{q^m}^l)\rtimes \Gal(\F_{q^m}/\F_p).
\]  
\end{proposition}

  \subsection{Matrix Code Equivalence Maps}\label{MatEquivMaps}
   In analogy with the notion of block-equivalence and rank-metric-equivalence, we now define equivalence for matrix codes.
    \begin{definition}\label{MatEquivalence}
 An invertible map $f:\Matlm \to \Matlm$ is a \emph{linear matrix-equivalence map} if $f$ is $\F_q$-linear and preserves rank weight.   The collection of linear matrix-equivalence maps is denoted by $\LEquiv_{\Mat}(\Matlm)$. Similarly, an invertible map $f:\Matlm \to \Matlm$ is a \emph{semi-linear matrix-equivalence map} if $f$ is $\F_q$-semi-linear and preserves rank weight.   The collection of semi-linear rank-metric-equivalence maps is denoted by $\SLEquiv_{\Mat}(\Matlm)$.  Two matrix codes $C, \tC \subseteq \Matlm$ are \emph{(semi-)linearly matrix-equivalent} if there exists a (semi-)linear matrix-equivalence map $f$ such that $\tC= f(C)$.
  \end{definition}
\begin{remark}
A rank-metric-equivalence map must be (semi-)linear with respect to the field over which the code is defined, namely the field $\F_{q^m}$, while matrix-equivalence maps are only required to be (semi-)linear with respect to the field $\F_q$, even though a common construction for matrix codes is the expansion of rank-metric codes that are linear over $\F_{q^m}$.  
\end{remark}
    \begin{remark}
Since the composition of two linear rank weight-preserving maps is also a linear rank weight-preserving map, the collection of linear equivalence maps forms a group under composition.  An analogous statement holds in the semi-linear case. 
  \end{remark}

To describe the collections of linear and semi-linear matrix-equivalence maps, we must first determine which maps preserve rank weight.   Proposition \ref{rankpreserving} below does precisely this for the case of linear maps.  To simplify the proof of this proposition, we begin with a lemma.

\begin{lemma}\label{rank1}
Let $\ve_i^{(n)}$ denote the $i^\text{th}$ standard basis vector of $\F_q^n$. 
For any $\vx \in \F_q^l$ and $\vy \in \F_q^m$ with $\vx, \vy \neq \vzero$, 
 if 
 \[
\rank\left( \ve_i^{(l)\top}\ve_j^{(m)} + \vx^\top\vy\right) =1,
 \]
then $\vx=\lambda \ve_i^{(l)}$ or $\vy = \lambda \ve_j^{(m)}$ for some $\lambda \in \F_q^*$.  
\end{lemma}
\begin{proof} 
Observe that 
\[
\ve_i^{(l)\top}\ve_j^{(m)} + \vx^\top\vy = {\scriptsize \begin{bmatrix} x_1 \vy \\ \vdots \\ \ve_j^{(m)}+x_i\vy \\ \vdots \\ x_l\vy \end{bmatrix}}.
\]
Suppose first that $x_i=0$, so that the $i^\text{th}$ row of $\ve_i^{(l)\top}\ve_j^{(m)} + \vx^\top\vy$ is just $\ve_j^{(m)}$.  Since $\vx \neq \vzero$ by hypothesis, there exists some $r\neq i$ such that $x_r\neq 0$, and so the $r^\text{th}$ row of $\ve_i^{(l)\top}\ve_j^{(m)} + \vx^\top\vy$ is $x_r\vy \neq \vzero$.  Thus, the rank 1 matrix  $\ve_i^{(l)\top}\ve_j^{(m)} + \vx^\top\vy$ has at least two non-zero rows, and there exists some $\lambda_r \neq 0$ such that 
\[
x_r\vy = \lambda_r \ve_j^{(m)}.
\]
Thus, $\vy=\frac{\lambda_r}{x_r}\ve_j^{(m)}$, i.e\ $\vy=\lambda \ve_j^{(m)}$ with $\lambda = \frac{\lambda_r}{x_r} \in \F_q^*$.

Now assume $x_i \neq 0$ and $\vy \neq \lambda \ve_j^{(m)}$ for any $\lambda \in \F_q^*$.  Then $\ve_j^{(m)} +x_i\vy \neq \vzero$, and so all other rows of $\ve_j^{(l)\top}\ve_j^{(m)} + \vx^\top\vy$ must be scalar multiples of $\ve_j^{(m)} +x_i\vy$.  In particular, for each $1\leq r \leq l$ with $r \neq i$, there exists some $\lambda_r \in \F_q$ such that 
\[
x_r \vy = \lambda_r\left(\ve_j^{(m)} +x_i\vy\right),
\]
and so $(x_r-\lambda_r x_i)\vy = \lambda_r \ve_j^{(m)}$.  Since we assumed $\vy \neq \lambda \ve_j^{(m)}$ for any $\lambda \in \F_q^*$, this equality can only hold if $\lambda_r=0$, which implies that $x_r=0$ for all $r \neq i$.  Hence, $\vx =\lambda \ve_i^{(l)}$ for some $\lambda \in \F_q^*$.  

\end{proof}

\begin{proposition}\label{rankpreserving} 
Let $f:\Matlm \to \Matlm$ be an invertible linear rank-preserving map.  Then there exist $L \in \GL_{l}(\F_q)$ and $M \in \GL_m(\F_q)$ such that either
\begin{itemize}
\item[] $f(A)=LAM$ for all $A \in \Matlm$ or
\item[] $f(A)=LA^\text{T} M$ for all $A \in \Matlm$,
\end{itemize}
where the latter case can only occur if $l = m$.
\end{proposition}
\begin{proof}
Without loss of generality, we may assume $l \leq m$; if $l>m$ we can apply the transpose, which is a rank-preserving map, and fall into the other case.  Let $f:\Matlm \to \Matlm$ be an invertible, linear, rank-preserving map, and let $E_{ij}=\ve_j^{(l)\top}\ve_j^{(m)}$ denote the $l \times m$ matrix of zeros with a 1 in the $(i,j)^\text{th}$ entry.  Since $E_{ij}$ has rank 1 and $f$ is rank-preserving,
$f(E_{ij})$ must also have rank 1.  Hence, $f(E_{ij})=\vx^\top \vy$ for
some $\vx\in \F_q^l$,  $\vy \in \F_q^m$.  For $1 \leq i \leq l$, let $\vx_i\in \F_q^l$, 
$\vy_i \in \F_q^m$ be such that $f(E_{ii}) = \vx_i^\top \vy_i.$  

Set $X:= {\footnotesize \begin{bmatrix} \vx_1^\top & \cdots & \vx_l^\top \end{bmatrix}} \in \F_q^{l \times l}$.  We claim that $X \in \GL_l(\F_q)$.  For any $\vv \in \F_q^m$, $\lambda_i(\vv):= \vy_i\vv^\top$ is an $\F_q$-scalar for $1 \leq i \leq l$, and so 
\begin{eqnarray*}
\Span\{\vx_i^\top~|~ 1 \leq i \leq l\} &\supseteq& \left\{\sum_{i=1}^l
 \vx_i^\top \lambda_i(\vv) ~|~ \vv \in \F_q^m \right\}= \left\{ \sum_{i=1}^l\vx_i^\top (\vy_i\vv^\top )  ~|~ \vv \in \F_q^m \right\}\\
&=& \left\{ \left(\sum_{i=1}^l f(E_{ii})\right)\vv^\top~|~ \vv \in \F_q^m \right\}=\left \{ f([I_l~|~0])\vv^\top~|~ \vv \in \F_q^m \right\}.
\end{eqnarray*}
Since $f$ is rank-preserving and each $f([I_l~|~0])\vv^\top$ is a linear combination of the columns of $f([I_l~|~0])$, we have 
\[
\dim \Span \{ f([I_l~|~0])\vv^\top~|~\vv \in \F_q^m\} = \dim \operatorname{colspan} \left( f([I_l~|~0]) \right) = \rank f([I_l~|~0])=l.
\]
Thus, since $\Span \{\vx_i^\top~|~ 1 \leq i \leq l\}\supseteq\Span \{ f([I_l~|~0])\vv^\top~|~\vv \in \F_q^m\}$ and $\dim \Span \{\vx_i~|~ 1 \leq i \leq l\} = \dim \Span \{\vx_i^\top~|~ 1 \leq i \leq l\}$, we have
\[
\dim \Span \{\vx_i~|~ 1 \leq i \leq l\} \geq \dim \Span \{ f([I_l~|~0])\vv^\top~|~\vv \in \F_q^m\} = l.
\]
Hence $\{\vx_i~|~ 1 \leq i \leq l\}$ is a linearly independent set, and so $X \in \GL_l(\F_q)$.  

Since $f': \Matml \to \Matml$ defined by $f'(A)=\left(f(A^\top )\right)^\top$ is also a rank-preserving map, and $f'(E_{ii})=\vy_i^\top\vx_i$, we may apply a similar argument to $f'$ to show that $\{\vy_i~|~1 \leq i \leq l\}$ is also a linearly independent set. Thus, if $l=m$, then $\{\vy_i~|~1 \leq i \leq l\}$ forms a basis for $\F_q^m$; if $l <m$ we arbitrarily extend $\{\vy_i~|~1 \leq i \leq l\}$ with vectors $\vy_{l+1} \ldots \vy_m$ to form a basis for $\F_q^m$, so that the matrix $Y:= {\scriptsize \begin{bmatrix} \vy_1 \\ \vdots \\ \vy_m\end{bmatrix}}$  is in $\GL_m(\F_q)$.  

By the definition of $X$ and $Y$, we see that $f(E_{ii}) = \vx_i^\top \vy_i = X E_{ii} Y$ for each $1 \leq i \leq l$.  Since the map $\hf: \Matlm \to \Matlm$ defined by $\hf(A) = X^{-1}f(A)Y^{-1}$ is also an invertible, linear, rank-preserving map and $\hf(E_{ii})= X^{-1}(X E_{ii} Y)Y^{-1}= E_{ii}$, we may assume, by replacing $f$ by $\hf$ if necessary, that the invertible, linear, rank-preserving map $f$ fixes each $E_{ii}$.



Since the map $f$ is completely determined by its action on $\{E_{ij} ~|~ 1 \leq i \leq l,~ 1\leq j \leq m\}$, all that remains is to understand $f(E_{ij})$ for all $i\neq j$.  First we will consider $1\leq i,j \leq l$ with $i \neq j$.  
Recall from above, that for each $i$ and $j$, there exist $\vx \in \F_q^l$ and $\vy \in \F_q^m$ such that $f(E_{ij})=\vx^\top\vy$.  Since $f$ is linear and rank-preserving and $f(E_{ii})=E_{ii}$, we have that
\[
f(E_{ii}+E_{ij}) = E_{ii} + f(E_{ij}) = \ve_i^{(l)\top}\ve_i^{(m)} + \vx^\top\vy
\]
must have rank 1.  Thus, by Lemma \ref{rank1}, $\vx= \lambda \ve_i^{(l)}$ or $\vy=\lambda \ve_i^{(m)}$ for some $\lambda \in \F_q^*$.  Similarly, since $f(E_{jj}+E_{ij})$ has rank 1, we also have that $\vx = \lambda' \ve_j^{(l)}$ or $\vy=\lambda' \ve_j^{(m)}$ for some $\lambda' \in \F_q^*$. Hence for $1 \leq i,j \leq l$ with $i \neq j$, we have
\[
\begin{array}{lll}
f(E_{ij}) = \lambda \lambda' \ve_i^{(l)\top}\ve_j^{(m)} = \lambda_{ij} E_{ij} & \text{or} &f(E_{ij}) = \lambda \lambda' \ve_j^{(l)\top}\ve_i^{(m)} = \lambda_{ij} E_{ji}
\end{array}
\]
where $\lambda_{ij} = \lambda \lambda'$.  

For $1 \leq i,j \leq l$ with $i \neq j$, we will say that $f$ is \emph{$ij$-preserving} if $f(E_{ij}) =\lambda_{ij}E_{ij}$; otherwise, $f(E_{ij})=\lambda_{ij}E_{ji}$ and we will say $f$
is \emph{$ij$-transposing}.  Since $f$ is also either $ik$-preserving or $ik$-transposing for each $1 \leq k \leq l$, we see that if $f$ is $ij$-preserving then 
\[
f(E_{ij} + E_{ik}) =\left\{\begin{array}{ll} \lambda_{ij} \ve_i^{(l)\top}\ve_j^{(m)} +  \lambda_{ik} \ve_i^{(l)\top}\ve_k^{(m)} & \text{if } f \text{ is } ik \text{-preserving}\\
\lambda_{ij} \ve_i^{(l)\top}\ve_j^{(m)} +  \lambda_{ik} \ve_k^{(l)\top}\ve_i^{(m)} & \text{if } f \text{ is } ik \text{-transposing} \end{array} \right.
\]
and so since $E_{ij} + E_{ik}$ has rank 1 and $f$ is rank-preserving we see from Lemma \ref{rank1} that $f(E_{ik})$ must equal $ \lambda_{ik} \ve_i^{(l)\top}\ve_k^{(m)}$, and so $f$ must be $ik$-preserving.  Applying a similar argument when $f$ is $ij$-transposing yields that $f$ is $ij$-preserving if and only if $f$ is $ik$-preserving.  Similarly, $f$ is $ij$-preserving if and only if $f$ is $kj$-preserving.  Thus, either
$f(E_{ij})=\lambda_{ij}E_{ij}$ for all $1 \leq i,j \leq l$, or
$f(E_{ij})=\lambda_{ij}E_{ji}$ for all $1 \leq i,j\leq l$, i.e.\ $f$ is
either \emph{leading-submatrix preserving} or \emph{leading-submatrix transposing}.  


Now we consider $l+1 \leq k \leq m$.  Fix $i$ with $1 \leq i \leq l$.  Applying the same argument as above to $f(E_{ii} + E_{ik})$, we have that $f(E_{ik})=\vw_k^\top\vz_k$ and $\vw_k=\lambda \ve_i^{(l)}$ or $\vz_k = \lambda \ve_i^{(m)}$.  For each $1 \leq j \leq l$,
\[
f(E_{ij}+E_{ik}) 
=\left\{\begin{array}{ll} \lambda_{ij}\ve_i^{(l)\top}\ve_j^{(m)} + \vw_k^\top\vz_k, & \text{if } f \text{ is leading-submatrix preserving}\\ 
\lambda_{ij} \ve_j^{(l)\top}\ve_i^{(m)}  + \vw_k^\top\vz_k, & \text{if } f \text{ is leading-submatrix transposing} 
\end{array}\right.
\]
and so by Lemma \ref{rank1}, 
\[
\begin{array}{ll} 
\vw_k=\lambda'\ve_i^{(l)} \text{ or } \vz_k=\lambda'\ve_j^{(m)}, &  \text{if } f \text{ is leading-submatrix preserving}\\ 
\vw_k=\lambda'\ve_j^{(l)} \text{ or } \vz_k=\lambda'\ve_i^{(m)}, & \text{if } f \text{ is leading-submatrix transposing}.
\end{array}
\]
Combining this with the previous constraints on $\vw_k$ and $\vz_k$, we see that for $l+1 \leq k \leq m$
\[
f(E_{ik}) 
=\left\{\begin{array}{ll} \lambda_{ik}\ve_i^{(l)\top}\vz_k, & \text{if } f \text{ is leading-submatrix preserving}\\ 
\lambda_{ik} \vw_k^\top\ve_i^{(m)}, & \text{if } f \text{ is leading-submatrix transposing}
\end{array}\right.
\]
for some $\lambda_{ik} \in \F_q^*$.  For $1 \leq k \leq l$, define $\vw_k:=\vx_k$ and $\vz_k:=\vy_k$, and by our earlier argument with $X:={\scriptsize \begin{bmatrix} \vx_1^\top \cdots \vx_l^\top \end{bmatrix}}$ and $Y:={\scriptsize \begin{bmatrix} \vy_1 \\ \vdots \\ \vy_m \end{bmatrix}}$, we may assume $\vx_k=\ve_k^{(l)}$ and $\vy_k=\ve_k^{(m)}$.  Since $f$ is linear and invertible,
\begin{eqnarray*}
m&=& \dim \Span \{E_{ik} ~|~1 \leq k \leq m\}= \dim \Span \{f(E_{ik}) ~|~1 \leq k \leq m\}\\
&=& \hspace{-.12in}\left\{ \hspace{-.05in}\begin{array}{ll} \dim \Span\{\ve_i^{(l)\top}\vz_{k}~|~1 \leq k \leq m\}, & \text{if } f \text{ is leading-submatrix preserving}\\ 
\dim \Span \{\vw_{k}^{\top}\ve_i^{(m)}~|~1 \leq k \leq m\}, & \text{if } f \text{ is leading-submatrix transposing}
\end{array}\right.\\
&=& \hspace{-.12in}\left\{ \hspace{-.05in}\begin{array}{ll} \dim \Span\{\vz_{k}~|~1 \leq k \leq m\}, & \text{if } f \text{ is leading-submatrix preserving}\\ 
\dim \Span \{\vw_{k}~|~1 \leq k \leq m\}, & \text{if } f \text{ is leading-submatrix transposing}, 
\end{array}\right.,
 \end{eqnarray*}
 where the final equality holds because, for a fixed $i$, both the maps $\ve_i^{(l)\top}\vz_{k} \mapsto \vz_k$ and $\vw_{k}^{\top}\ve_i^{(m)} \mapsto \vw_k$ are linear bijections.  Since $\vw_{k} \in \F_q^l$, $\dim \Span\{\vw_{k}~|~1 \leq k \leq m\}\leq l$.  Thus, if $l < m$, we cannot have $\dim \Span\{\vw_{k}~|~1 \leq k \leq m\} =m$, and so $f$ cannot be leading-submatrix transposing.  Thus, if $l<m$, then $f$ is leading-submatrix preserving and, since $m= \dim \Span \{\vz_{k}~|~1 \leq k \leq m\}$, we have that $\{\vz_{k}~|~1 \leq k \leq m\}$ is a linearly independent set.  Thus, the matrix $Z :={\scriptsize \begin{bmatrix} \vz_{1} \\ \vdots \\ \vz_m \end{bmatrix}}$ is in $\GL_m(\F_q)$.  

If $l < m$, then for $1 \leq i \leq l$ and $1 \leq k \leq m$, $f(E_{ik}) = \lambda_{ik} \ve_i^{(l)\top}\vz_k$.  Since the map $\hat{\hat{f}}: \Matlm \to \Matlm$ defined by $\hat{\hat{f}}(A) = f(A)Z^{-1}$ is also an invertible, linear, rank-preserving map and $\hat{\hat{f}}(E_{ik})= \lambda_{ik} \ve_i^{(l)\top}\vz_k = (\lambda_{ik} E_{ik} Z)Z^{-1} = \lambda_{ik} E_{ik}$, we may assume, by replacing $f$ by $\hat{\hat{f}}$ if necessary, that if $l <m$, then the invertible, linear, rank-preserving map $f$ satisfies $f(E_{ik})=\lambda_{ik} E_{ik}$ for $1 \leq i \leq l$, $1 \leq k \leq m$.  Recall that if $l=m$, then for $1\leq i,j \leq l$, we have $f(E_{ij}) = \lambda_{ij} E_{ij}$ or $f(E_{ij}) = \lambda_{ij} E_{ji}$.  Thus, in either case, all that remains to understand $f$ is to determine the values of $\lambda_{ij}$ for $1 \leq i \leq l$ and $1 \leq j \leq m$.  

We now show that without loss of generality, we may assume that $\lambda_{ij}=1$ for all $1 \leq i \leq l$ and $1 \leq j \leq m$.  
For ease of notation, we will only consider the case when $f$ is leading-submatrix preserving, but a similar argument holds when $f$ is leading-submatrix transposing. 
Under the assumption that $f$ is leading-submatrix preserving, $f(E_{ij})=\lambda_{ij}E_{ij}$ for $1 \leq i\leq l$, $1 \leq j \leq m$, and so $f\left(\sum_{i=1}^l \sum_{j=1}^m E_{ij}\right)= \Lambda$ where $\Lambda = [\lambda_{ij}]$.  Let $\vlambda_i=(\lambda_{i1}, \lambda_{i2}, \ldots, \lambda_{im})$ denote the $i^\text{th}$ row of $\Lambda$.  
Since $\sum_{i=1}^l \sum_{j=1}^m E_{ij}$ has rank 1 and $f$ is rank-preserving, $\Lambda$ must have rank 1.  Recall that by hypothesis, $f(E_{ii})=E_{ii}$, and so $\lambda_{ii}=1$ for $1 \leq i \leq l$, and in particular, $\lambda_{11}=1$, and so $\vlambda_1$ is non-zero.  Thus, since $\Lambda$ has rank 1, there exists some $\alpha_i \in \F_q$ such that $\vlambda_i = \alpha_i \vlambda_1$ for each $1 \leq i \leq l$.  Using the fact that $\lambda_{jj}=1$, we see that each $\alpha_j \neq 0$ and $\lambda_{1j}=\alpha_j^{-1}$.    Hence, for $1 \leq i,j \leq l$, $\lambda_{ij} = \alpha_i \lambda_{1j} = \alpha_i \alpha_j^{-1}$. 
For $l+1 \leq k \leq m$, set $\beta_k=\lambda_{1k}$.  Then since $\Lambda$ has rank 1 and since the first column of $\Lambda$ is $(1, \alpha_2, \ldots, \alpha_l)^\top$, we have that $\lambda_{ik}=\alpha_i \beta_k.$

Let $D_1=\text{diag}(1, \alpha_2, \ldots, \alpha_l) \in \GL_l(\F_q)$ and let $D_2=\text{diag}(1, \alpha_2^{-1}, \ldots, \alpha_l^{-1},\beta_{l+1}, \ldots, \beta_m)\in\GL_m(\F_q)$.  Then for $1 \leq i \leq l$ and $1 \leq j \leq m$, since 
\[
f(E_{ij})= \lambda_{ij} E_{ij} =  \left\{\begin{array}{ll} \alpha_i \alpha_j^{-1} E_{ij}, & \text{if } 1 \leq j \leq l\\
\alpha_i \beta_j E_{ij},& \text{if } l+1 \leq j \leq m
\end{array} \right.
\]
we have $f(E_{ij})= D_1 E_{ij} D_2$. Since the map $\tf: \Matlm \to \Matlm$ defined by $\tf(A) = D_1^{-1}f(A)D_2^{-1}$ is also an invertible, linear, rank-preserving map and $\tf(E_{ij})= D_1^{-1}(D_1 E_{ij} D_2)D_2^{-1}= E_{ij}$, we may assume, by replacing $f$ by $\tf$ if necessary, that the invertible, linear, rank-preserving map $f$ fixes each $E_{ij}$, and so each $\lambda_{ij}=1$ for $1 \leq i\leq l$, $1 \leq j \leq m$.

Finally, since the images of $E_{ij}$, $1 \leq i,j \leq l$ completely determine $f$,
 we have (under the assumption that $f(E_{ij})=E_{ij}$ or if $l=m$, possibly $f(E_{ij}) = E_{ji}$) that either $f(A)=A$ for all $A \in \Matlm$ or $f(A)=A^\top $ for all $A \in \Matlm$, where the second case can only occur if $l=m$.  Let $L=D_1X$ and $M=YZD_2$ where $X,Y, Z, D_1,$ and $D_2$ are the matrices defined above that enabled the assumption $f(E_{ij})=E_{ij}$ or $f(E_{ij})=E_{ji}$.  Then we see every invertible, linear, rank-preserving map $f$ is either of the form $f(A) = LAM$ for every $A \in \Matlm$ or of the form $f(A)=LA^\top M$ for every $A \in \Matlm$, again where the second case can only occur if $l=m$.
\end{proof}

From Proposition \ref{rankpreserving}, we see that the collection of linear matrix-equivalence maps consists of only those maps that are a composition of multiplication on the left or right by invertible matrices and, when the matrix codewords are square, matrix transposition.  To determine the group structure of this collection, we must recast these maps so that they live within some common group.  Since they are all linear maps acting on objects with $l\times m$ entries, we may switch to viewing these maps as elements of $\GL_{lm}(\F_q)$ acting on extended row vectors of length $lm$ instead of on $l \times m$ matrices, where these vectors are formed simply by concatenating the $l$ rows of the matrix.  We will denote the collection of \emph{matrix equivalence maps acting on extended row vectors} by $\LEquiv_{\Vector}(\Matlm)$.  As a subgroup of $\GL_{lm}(\F_q)$, it is relatively straight forward to show that $\LEquiv_{\Vector}(\Matlm)$ has the structure of a semi-direct product of the subgroup generated by the map of order two for matrix transposition and the subgroup generated by matrices of the form $L \otimes M$ where $L \in \GL_l(\F_q)$, $M \in \GL_m(\F_q)$ and where $\otimes$ denotes the Kronecker product.  Additionally, since $\lambda I_l \otimes \lambda^{-1}I_m = I_l \otimes I_m =I_{lm}$ for any $\lambda \in \F_q^*$, we see that the subgroup generated by matrices of the form $L \otimes M$ is isomorphic to $\GL_l(\F_q) \times \GL_l(\F_q)$ modded out by the subgroup generated by $\lambda I_l \otimes \lambda^{-1}I_m$.  As a consequence of this result and the fact that there is a natural isomorphism between $\LEquiv_{\Vector}(\Matlm)$ and $\LEquiv_{\Mat}(\Matlm)$, we have the following proposition; for details of this proof, we refer the reader to the commentary after Proposition 3.2.17 through Corollary 3.2.22 in \cite{MyThesis}.  

\begin{proposition}\label{LEquivMat}
There is an isomorphism of groups 
\[
\LEquiv_{\Mat}(\Matlm) \cong \LEquiv_{\Vector}(\Matlm) \cong \left\{ \begin{array}{ll} \Z_2 \ltimes \left(\GL_l(\F_q) \times \GL_l(\F_q)\right)/N & \text{if } l=m \\ 

\left(\GL_l(\F_q) \times \GL_m(\F_q)\right)/N & \text{if } l \neq m 
\end{array} \right.
\]
where $N=\{(\lambda I_l, \lambda^{-1}I_m)~|~\lambda~\in~\F_q^*\} \leq \GL_l(\F_q) \times \GL_m(\F_q)$. 
\end{proposition}
\begin{remark}
Again, we see each linear matrix equivalence map corresponds to a coset of the form $(L,M) \cdot N$ for some $L \in \GL_l(\F_q)$ and $M \in \GL_m(\F_q)$.  For ease of notation, we will henceforth write $[L,M]$ to denote the coset $(L,M) \cdot N$ as before.  
\end{remark}

 \begin{proposition}\label{SLEquivMat}
 The group $\SLEquiv_{\Mat}(\Matlm)$ of semi-linear matrix-equivalence maps is given by
\[
\SLEquiv_{\Mat}(\Matlm) \cong\LEquiv_{\Mat}(\Matlm) ~\rtimes~\Gal(\F_q/\F_p).
\]
 \end{proposition}
\begin{proof}
 Let $\gamma\in \Gal(\F_{q}/\F_p)$ and $X,Y \in \Matlm$ with $\vx_1, \ldots, \vx_l$ denoting the rows of $X$ and $\vy_1, \ldots, \vy_l$ denoting the rows of $Y$.  Then
   \begin{eqnarray*}
  d_R(X^\gamma, Y^\gamma) &=&  \rank(X^\gamma-Y^\gamma) \\
  &=&  \text{dim}_{\F_q} \left(\text{span}_{\F_q} \{\vx_1^\gamma - \vy_1^\gamma, \ldots,~ \vx_l^\gamma-\vy_l^\gamma \} \right)\\
&=&  \text{dim}_{\F_q} \left(\text{span}_{\F_q} \{\vx_1 - \vy_1, \ldots,~ \vx_l-\vy_l \} \right)^\gamma\\
&=&  \text{dim}_{\F_q} \left(\text{span}_{\F_q} \{\vx_1 - \vy_1, \ldots,~ \vx_l-\vy_l \} \right)\\
&=& \rank(X-Y)\\
&=& d_R(X,Y)
  \end{eqnarray*}
  where the fourth equality holds because $\gamma$ corresponds to a vector space automorphism of $ \F_{q}^m$, and so it will map subspaces to other subspaces of the same dimension.  Thus, each automorphism is a rank weight preserving map, and so $\{(I_l; \gamma)~|~\gamma \in \Gal(\F_{q}/\F_p)\} \subseteq \SLEquiv_{\Mat}(\Matlm)$.  By Proposition \ref{LEquivMat}, 
$  \{(A; id)~|~ A~\in~\LEquiv_{\Mat}(\Matlm)\}~\subseteq~\SLEquiv_{\Mat}(\Matlm)$ as well, and so 
  \[
  \LEquiv_{\Mat}(\Matlm)\rtimes \Gal(\F_{q}/\F_p) \subseteq \SLEquiv_{\Mat}(\Matlm).
  \]
  
  To show reverse containment, let $(A; \gamma)$ be an arbitrary element of $\SLEquiv_{\Mat}(\Matlm)$.  By the argument above $(I_l; \gamma^{-1}) \in \SLEquiv_{\Mat}(\Matlm)$, and since $\SLEquiv_{\Mat}(\Matlm)$ is a group under composition, we have that $(A; \gamma)(I_l;\gamma^{-1}) = (A; id)$ is a semi-linear matrix equivalence map as well.  But $(A; id)$ is in fact a linear map, and so it must be a linear matrix equivalence map.  Thus, $A \in \LEquiv_{\Mat}(\Matlm)$, and so $(A; \gamma) \in \LEquiv_{\Mat}(\Matlm) \rtimes \Gal(\F_{q^m}/\F_p)$.  
%
\end{proof}

\subsection{Relationship between Rank-Metric and Matrix Code\\ Equivalence Maps}
Recall from Section \ref{background} that to each rank-metric code we may associate an matrix code via the map $\epsB$ for matrix expansion with respect to an ordered basis, and to each matrix code we may associate a rank-metric code via the map $\epsB^{-1}$ for matrix compression.  This association provides a natural framework for comparing rank-metric and matrix code equivalence to determine if, for example, there are codes that would be viewed as equivalent in one setting while being viewed as inequivalent in the other.  

In Theorem \ref{LEquivRMLEquivMat}, we will show that the notion of linear matrix-equivalence is strictly more general than the notion of linear rank-metric-equivalence.  Specifically, we will show that whenever two rank-metric codes are linearly rank-metric equivalent, their matrix expansions are always linearly matrix equivalent, but the converse is only true under certain conditions.  In Theorem \ref{SLEquivRMSLEquivMat} below, we show a similar result for semi-linear rank-metric- and matrix-equivalence.  First we need a lemma characterizing the matrix representations of the $\F_q$-linear transformations for multiplication by $\alpha \in \F_{q^m}$ and for $q$-exponentiation, as well as a lemma characterizing how the map $\epsB$ interacts with matrix multiplication and with $\Gal(\F_q/\F_p)$.
  
\begin{lemma}\label{M_alphaQ}
Let $\vb=(b_1, \ldots, b_m)$ be an ordered basis for $\F_{q^m}$ over $\F_q$. 
Then for any $\vx \in \F_{q^m}^l$ and any $\alpha\in \F_q^*$, we have  
\[
\begin{array}{lll}
\epsB(\alpha\vx)=\epsB(\vx)M_{\alpha} &\text{and} &\epsB(\vx^q)=\epsB(\vx)Q,
\end{array}
\]
where $M_{\alpha}=\epsB(\alpha \vb)$ and $Q=\epsB(\vb^q)$.  In other words, $M_{\alpha}$ is the matrix for the $\F_q$-linear transformation of multiplication by $\alpha$ and $Q$ is the matrix for the $\F_q$-linear transformation of exponentiation by $q$. 
\end{lemma}

\begin{proof}
For any $X \in \Matlm$, we have $\epsB^{-1}(X) = (X (b_1, \ldots, b_m)^\top)^\top = (b_1, \ldots, b_m) X^\top$ from Definition \ref{compression}.  Hence, for $\vx \in \F_{q^m}^l$ with $\epsB(\vx) = [x_{ij}]$, 
\begin{eqnarray*}
\epsB^{-1}\left(\epsB(\vx)M_{\alpha}\right) &=& (b_1, \ldots, b_m) \left(\epsB(\vx)M_{\alpha}\right)^\top\\
&=&  (b_1, \ldots, b_m) M_{\alpha}^\top \epsB(\vx)^\top\\
&=& (\alpha b_1, \ldots, \alpha b_m) \epsB(\vx)^\top\\
&=& \left(\sum_{j=1}^m x_{1j} \alpha b_j, \ldots, \sum_{j=1}^m x_{lj} \alpha b_j\right)\\
&=&  \alpha \left(\sum_{j=1}^m x_{1j} b_j, \ldots, \sum_{j=1}^m x_{lj} b_j\right)\\
&=& \alpha \vx,
\end{eqnarray*}
where the third equality holds because the rows of $M_{\alpha}$ are the images of the basis elements under multiplication by $\alpha$.  Applying $\epsB$ to both sides of the equation yields the first result.  Since the rows of $Q$ are the images of the basis elements under $q$-exponentiation, a similar argument yields the second result.
 \end{proof}

\begin{lemma}\label{epsBL}
Let $\vb= (b_1, \ldots, b_m) \subseteq \F_{q^m}$ be an ordered basis for $\F_{q^m}$ over $\F_q$.  Let $\sigma_p$ be the Frobenius automorphism, and let $1 \leq r \leq e-1$.  Then for any $\vx \in \F_{q^m}^l$ 
\[
\begin{array}{lll}
\epsB(\vx L) = L^\top \epsB(\vx) & \text{and} &\epsB\left(\vx^{\sigma_p^r}\right)=(\epsB(\vx)P_r)^{\sigma_p^r},
\end{array}
\]
where $L \in \GL_l(\F_q)$ and $P_r=\left(\epsB(\vb^{\sigma_p^r})\right)^{\sigma_p^{-r}}$.
\end{lemma}
\begin{proof}
Using Definition \ref{compression}, we have
\begin{eqnarray*}
\epsB^{-1}(L^\top \epsB(\vx)) &=& (b_1, \ldots, b_m) (\epsB(\vx))^\top L \\
&=&\epsB^{-1}( \epsB(\vx) )L\\
&=& \vx L.
\end{eqnarray*}
Applying $\epsB$ to both sides of the previous equation, we obtain the desired first result.

Again from Definition \ref{compression}, we see
\begin{eqnarray*}
\epsB^{-1}\left((XP_r)^{\sigma_p^r}\right) &=&  (b_1, \ldots, b_m) \left(X^{\sigma_p^r}P_r^{\sigma_p^r}\right)^\top\\
&=& (b_1, \ldots, b_m) \left(P_r^{\sigma_p^r}\right)^\top \left(X^{\sigma_p^r}\right)^\top\\
&=&(b_1, \ldots, b_m)  \epsB\left(\vb^{\sigma_p^r}\right)\top \left(X^{\sigma_p^r}\right)^\top\\
&=& (b_1^{p^r}, \ldots, b_m^{p^r}) \left(X^{\sigma_p^r}\right)^\top\\
&=& \left(\sum_{j=1}^m x_{1j}^{\sigma_p^r}b_j^{\sigma_p^r}, \ldots, \sum_{j=1}^m x_{lj}^{\sigma_p^r}b_j^{\sigma_p^r}\right)  \\
&=& \left(\left(\sum_{j=1}^m x_{1j}b_j\right)^{\sigma_p^r}, \ldots, \left(\sum_{j=1}^m x_{lj}b_j\right)^{\sigma_p^r}\right)  \\
&=& \left(\sum_{j=1}^m x_{1j}b_j, \ldots, \sum_{j=1}^m x_{lj}b_j\right)^{\sigma_p^r}\\
&=& \vx^{\sigma_p^r}.
\end{eqnarray*}
We obtain the second result by applying $\epsB$ to both sides of the equation above.
\end{proof}

\begin{theorem}\label{LEquivRMLEquivMat}
Let $C_1,C_2 \subseteq \F_{q^m}^l$ be rank-metric codes.  If $C_1$ and $C_2$ are linearly rank-metric equivalent, then $\epsB(C_1)$ and $\epsB(C_2)$ are linearly matrix equivalent for any basis $\vb$ of $\F_{q^m}$ over $\F_q$.

Conversely, for a fixed basis $\vb$ of $\F_{q^m}$ over $\F_q$, if $\epsB(C_1)$ and $\epsB(C_2)$ are linearly matrix equivalent, then $C_1$ and $C_2$ are linearly rank-metric equivalent if and only if there exists a map $g \in \LEquiv_{\Mat}(\Matlm)$ that satisfies $g(\epsB(C_1))= \epsB(C_2)$ and has the form 
\[
g(A) = LAM_{\alpha} \text{ for all }A \in \Matlm
\]
for some $L \in \GL_l(\F_q)$ and some $M_{\alpha}$ as in Lemma \ref{M_alphaQ}.  
\end{theorem}
\begin{remark}
Note that in the second portion of the statement, we only assert that there exists a $g \in \LEquiv_{\Mat}(\Matlm)$ of the specified form, but we do not assert that every linear matrix equivalence map sending $\epsB(C_1)$ to $\epsB(C_2)$ will have that form.  The reason for this distinction is that given the map $g$ from the statement, one can compose it with any linear matrix equivalence map fixing $\epsB(C_1)$ and that composition will map $\epsB(C_1)$ to $\epsB(C_2)$; however, not every equivalence map that fixes $\epsB(C_1)$ will have the form specified in the statement, and so the composition need not have the desired form.  
\end{remark}
\begin{proof}
Let $C_1,C_2 \subseteq \F_{q^m}^l$ be linearly rank-metric equivalent codes and fix $f \in \LEquiv_{\RM}(\F_{q^m}^l)$ with $f(C_1) = C_2$. By Proposition \ref{LEquivRM}, $f$ has a representative in $\left(\F_{q^m}^* \times \GL_l(\F_q)\right) /N$ of the form $[\alpha, L]$ for some $\alpha \in \F_{q^m}^*$ and $L \in \GL_l(\F_q)$, where $N = \{ (\lambda,~\lambda^{-1} I_l) ~|~ \lambda \in \F_q^*\}$.  Hence, $C_2 = \alpha C_1 L:= \{ \alpha \vx L ~|~ \vx \in C_1\}$.  Let $\vb$ be an arbitrary ordered basis for $\F_{q^m}$ over $\F_q$.  Then
\[
\begin{array}{lcll}
\epsB(C_2) &=& \epsB(\alpha C_1 L)& \\
&=&\epsB(C_1 L) M_{\alpha} & \text{by Lemma \ref{M_alphaQ}} \\
&=& L^\top \epsB(C_1) M_{\alpha} & \text{by Lemma \ref{epsBL}}.
\end{array}
\]
Define the map $g: \Matlm \to \Matlm$ by $g(A) = L^\top A M_{\alpha}$ for all $A \in \Matlm$.  Since $L^\top \in \GL_l(\F_q)$ and $M_{\alpha} \in \GL_m(\F_q)$, we have $g \in \LEquiv_{\Mat}(\Matlm)$ by Proposition \ref{rankpreserving}, and $g(\epsB(C_1)) = \epsB(C_2)$.  Hence $\epsB(C_1)$ and $\epsB(C_2)$ are linearly matrix equivalent for any basis $\vb$ of $\F_{q^m}$ over $\F_q$, and so the first result holds.

Now assume that $\epsB(C_1)$ and $\epsB(C_2)$ are linearly matrix equivalent where $\vb$ is some fixed basis for $\F_{q^m}$ over $\F_q$.  If, in addition, $C_1$ and $C_2$ are linearly rank-metric equivalent, then there exists some $f \in \LEquiv_{\RM}(\F_{q^m}^l)$ with $f(C_1) = C_2$.  As above, $f$ has a representative in $\left(\F_{q^m}^* \times \GL_l(\F_q)\right) /N$ of the form $[\alpha, L]$ for some $\alpha \in \F_{q^m}^*$ and $L \in \GL_l(\F_q)$, and so by the same logic as above
\[
\epsB(C_2) =L^\top \epsB(C_1) M_{\alpha}.
\]
Define $g: \Matlm \to \Matlm$ by $g(A) = L^\top A M_{\alpha}$ for all $A \in \Matlm$.  Then $g \in \LEquiv_{\Mat}(\Matlm)$ and $g$ satisfies the conditions of the statement of the theorem, and so the result holds.  

Conversely, suppose that there exists some $g \in  \LEquiv_{\Mat}(\Matlm)$ with $g(\epsB(C_1))=\epsB(C_2)$ such that $g$ has the form $g(A) = L A M_{\alpha}$ for all $A \in \Matlm$ for some $L \in \GL_l(\F_q)$ and $\alpha \in \F_{q^m}^*$.  Then
\begin{eqnarray*}
C_2&=& \epsB^{-1}( \epsB(C_2))\\
&=& \epsB^{-1}( g(\epsB(C_1)))\\
&=& \epsB^{-1}( L \epsB(C_1)M_{\alpha})\\
&=&  \epsB^{-1}( \epsB(\alpha C_1 L^\top))\\
&=& \alpha C_1 L^\top
\end{eqnarray*}

Define $f: \F_{q^m}^l \to \F_{q^m}^l$ by $f(\vx) = \alpha \vx L^\top$ for all $\vx \in \F_{q^m}^l$.  Then $f \in \LEquiv_{\RM}(\F_{q^m}^l)$ by Proposition \ref{LEquivRM} and $f(C_1)=C_2$.  Hence $C_1$ and $C_2$ are linearly rank-metric equivalent, and so the second result holds.
\end{proof}

In Theorem \ref{SLEquivRMSLEquivMat} below, we prove a similar result for semi-linear rank-metric and matrix equivalence, but first we need a lemma characterizing the subgroup formed by the matrices for the $\F_q$-linear transformations of $\F_{q^m}$-scalar multiplication and $q$-exponentiation.  

\begin{lemma}\label{Ksubgp}
Let $\vb=(b_1, \ldots, b_m)$ be an ordered basis for $\F_{q^m}$ over $\F_q$.  Let $\alpha$ be a primitive element for $\F_{q^m}$ and let $M_{\alpha}$ and $Q$ be as in Lemma \ref{M_alphaQ}.  Define the subset $K \subseteq \GL_m(\F_q)$ by $K=\lp M_{\alpha} \rp \cdot \lp Q \rp$.  Then $K$ is a subgroup of $\GL_m(\F_q)$ with $K \cong \lp M_{\alpha} \rp \rtimes \lp Q \rp$ and $|K|=m(q^m-1)$. 
\end{lemma}
\begin{proof}
First we show $M_{\alpha} Q= QM_{\alpha}^{q} $.  Let $X \in \GL_m(\F_q)$ and set $\vx = \epsB^{-1}(X)$.  Repeated application of Lemma \ref{M_alphaQ} yields
\begin{eqnarray*}
XM_\alpha Q &=& \epsB(\vx)M_\alpha Q \\
&=& \epsB(\alpha \vx) Q\\
&=& \epsB\left((\alpha \vx)^{q}\right)\\
&=& \epsB\left(\alpha^q \vx^{q}\right)\\
&=& \epsB\left(\vx^{q}\right) M_{\alpha^q} \\
&=& \epsB(\vx)QM_\alpha^q\\
&=&XQM_\alpha^q.
\end{eqnarray*}
Multiplying both sides of the equation by $X^{-1}$ yields the desired result.

Hence $\lp Q \rp$ is contained in the normalizer of $\lp M_{\alpha} \rp$, and so $K$ is a subgroup of $\GL_m(\F_q)$ with
\[
|K| =\displaystyle \frac{| \lp M_{\alpha} \rp||\lp Q \rp|}{|\lp M_{\alpha} \rp\cap\lp Q \rp|}.
\]
Since $\epsB(\alpha \vx)Q^i =\epsB(\alpha^{q^i}\vx^{q^i}) \neq \epsB(\alpha\vx^{q^i})$, each $Q^i$ corresponds to a map on $\F_{q^m}^l$ that is merely $\F_{q^m}$-semi-linear, while each $M_{\alpha}^i$ corresponds to an $F_{q^m}$-linear map, and so we have that  $\lp M_{\alpha}\rp\cap\lp Q \rp = \{I_m\}$.  Thus, $|K|=m(q^m-1)$. Furthermore, since $\lp M_{\alpha} \rp $ and $\lp Q \rp$ have trivial intersection, we see that $K$ is an internal semi-direct product $ \lp M_{\alpha} \rp \rtimes \lp Q \rp$ with multiplication defined by 
\[
(M_{\alpha}^{i_1}; Q^{j_1})(M_{\alpha}^{i_2}; Q^{j_2}) = (M_{\alpha}^{i_1}M_{\alpha}^{i_2q^{j_2}}; Q^{j_1+j_2}).
\]
\end{proof}

\begin{theorem}\label{SLEquivRMSLEquivMat}
Let $C_1,C_2 \subseteq \F_{q^m}^l$ be rank-metric codes.  If $C_1$ and $C_2$ are semi-linearly rank-metric equivalent, then $\epsB(C_1)$ and $\epsB(C_2)$ are semi-linearly matrix equivalent for any basis $\vb$ of $\F_{q^m}$ over $\F_q$.

Conversely, for a fixed basis $\vb$ of $\F_{q^m}$ over $\F_q$, if $\epsB(C_1)$ and $\epsB(C_2)$ are semi-linearly matrix equivalent, then $C_1$ and $C_2$ are semi-linearly rank-metric equivalent if and only if there exists a map $g \in \SLEquiv_{\Mat}(\Matlm)$ that satisfies $g(\epsB(C_1))= \epsB(C_2)$ and has the form 
\[
g(A) = (LAMP_r)^{\sigma_p^r} \text{ for all }A \in \Matlm
\]
for some $L \in \GL_l(\F_q)$, $M \in K$, and $1 \leq r \leq e$, where $P_r$ and $\sigma_p$ are as in Lemma \ref{epsBL} and $K$ is as in Lemma \ref{Ksubgp}.
\end{theorem}
\begin{proof}
Let $C_1,C_2 \subseteq \F_{q^m}^l$ be semi-linearly rank-metric equivalent codes.  Then there exists some $f \in \SLEquiv_{\RM}(\F_{q^m}^l)$ with $f(C_1) = C_2$. By Proposition \ref{SLEquivRM}, $f$ has a representative in $\left(\left(\F_{q^m}^* \times \GL_l(\F_q)\right) /N\right) \rtimes \Gal(\F_{q^m} / \F_p)$ of the form $([\alpha, L]; \gamma)$ for some $\alpha \in \F_{q^m}^*$, $L \in \GL_l(\F_q)$, and $\gamma \in \Gal(\F_{q^m} / \F_p)$, where $N = \{ (\lambda,~\lambda^{-1} I_l) ~|~ \lambda \in \F_q^*\}$, and we have $C_2 = (\alpha C_1 L)^\gamma:= \{ (\alpha \vx L)^\gamma ~|~ \vx \in C_1\}$.  Since $\Gal(\F_{q^m} / \F_p) = \{ \sigma_p^i~|~1 \leq i \leq me\}$, there exists some $1 \leq i \leq me$ such that $\gamma = \sigma_p^i$; write $i = ej+r$ with $0 \leq r \leq e-1$, so that
\[
\gamma =  \sigma_p^{ej+r} = (\sigma_p^e)^j\sigma_p^r = (\sigma_{p^e})^j\sigma_p^r = \sigma_q^j\sigma_p^r.
\]
Let $\vb$ be an arbitrary basis for $\F_{q^m}$ over $\F_q$.  Then
\[
\begin{array}{lcll}
\epsB(C_2) &=& \epsB((\alpha C_1 L)^\gamma)& \\
&=& \epsB((\alpha C_1 L)^{\sigma_q^j\sigma_p^r})& \\
&=& \left(\epsB((\alpha C_1 L)^{\sigma_q^j})P^r\right)^{\sigma_p^r} & \text{by Lemma \ref{epsBL}}\\
&=& \left(\epsB(\alpha C_1 L)Q^jP^r\right)^{\sigma_p^r} & \text{by Lemma \ref{M_alphaQ}}\\
&=&\left(\epsB(C_1 L) M_{\alpha}Q^jP^r\right)^{\sigma_p^r}  & \text{by Lemma \ref{M_alphaQ}} \\
&=& \left(L^\top \epsB(C_1) M_{\alpha}Q^jP^r\right)^{\sigma_p^r}  & \text{by Lemma \ref{epsBL}}.
\end{array}
\]
Define $g: \Matlm \to \Matlm$ by $g(A) = (L^\top A M_{\alpha}Q^jP_r)^{\sigma_p^r}$ for all $A \in \Matlm$.  Since $L^\top \in \GL_l(\F_q)$, $M_{\alpha}Q^jP_r \in \GL_m(\F_q)$, and $\sigma_p^r \in \Gal(\F_q/\F_p)$, we have that $g \in \SLEquiv_{\Mat}(\Matlm)$ by Propostion \ref{SLEquivMat}.  Thus, $g$ is a semi-linear matrix equivalence map and $g(\epsB(C_1)) = \epsB(C_2)$.  Hence $\epsB(C_1)$ and $\epsB(C_2)$ are semi-linearly matrix equivalent for any basis $\vb$ of $\F_{q^m}$ over $\F_q$, and so the first result holds.

Now assume that $\epsB(C_1)$ and $\epsB(C_2)$ are semi-linearly matrix equivalent where $\vb$ is some fixed basis for $\F_{q^m}$ over $\F_q$.  If, in addition, $C_1$ and $C_2$ are semi-linearly rank-metric equivalent, then there exists some $f \in \SLEquiv_{\RM}(\F_{q^m}^l)$ with $f(C_1) = C_2$.  As above, 
$f$ has a representative in $\left(\left(\F_{q^m}^* \times \GL_l(\F_q)\right) /N\right) \rtimes \Gal(\F_{q^m} / \F_p)$ of the form $([\alpha, L]; \gamma)$ for some $\alpha \in \F_{q^m}^*$, $L \in \GL_l(\F_q)$, and $\gamma = \sigma_q^j\sigma_p^r \in \Gal(\F_{q^m}/\F_p)$, and so by the same reasoning as above
\[
\epsB(C_2) = \left(L^\top \epsB(C_1) M_{\alpha}Q^jP^r\right)^{\sigma_p^r}.
\]
Define $g: \Matlm \to \Matlm$ by $g(A) = (L^\top A M_{\alpha}Q^jP_r)^{\sigma_p^r}$ for all $A \in \Matlm$.  Then $g \in \SLEquiv_{\Mat}(\Matlm)$. Since  $(M_{\alpha}Q^jP_r )P_r^{-1} = M_{\alpha}Q^j \in K$, $g$ satisfies the conditions of the statement of the theorem, and so the result holds.  

Conversely, suppose that there exists some $g \in  \LEquiv_{\Mat}(\Matlm)$ with $g(\epsB(C_1))=\epsB(C_2)$ such that $g$ has the form $g(A) = (L A MP_r)^{\sigma_p^r}$ for all $A \in \Matlm$ for some $L \in \GL_l(\F_q)$, $M \in K$, and $1\leq r \leq e$. Since $M \in K$ and $K = \lp M_{\alpha} \rp \rtimes \lp Q \rp$, we can write $M=M_{\alpha}^i Q^j$ for some $1 \leq i \leq q^m-1$ and $ 1 \leq j \leq m$, and we have
\begin{eqnarray*}
\epsB(C_2) &=& g(\epsB(C_1)) \\ 
&=&(L \epsB(C_1)M P_r)^{\sigma_p^r}\\
&=&(L \epsB(C_1)M_{\alpha}^iQ^j P_r)^{\sigma_p^r}\\ 
&=&( \epsB(\alpha^iC_1L^\top)Q^j P_r)^{\sigma_p^r}\\  
&=& \epsB\left((\alpha^iC_1L^\top)^{\sigma_q^j}\right) P_r)^{\sigma_p^r}\\  
&=&\epsB\left((\alpha^iC_1L^\top)^{\sigma_q^j\sigma_p^r}\right)\\  
&=& \epsB\left( (\alpha^iC_1L^\top)^{\sigma_p^{ej+r}}\right).   
\end{eqnarray*}

Define $f: \F_{q^m}^l \to \F_{q^m}^l$ by $f(\vx) = (\alpha^i \vx L^\top)^{\sigma_p^{ej+r}}$ for all $\vx \in \F_{q^m}^l$.  Then $f \in \SLEquiv_{\RM}(\F_{q^m}^l)$ by Proposition \ref{SLEquivRM} and, applying $\epsB^{-1}$ to both sides of the previous equality, we see $f(C_1)=C_2$.  Hence $C_1$ and $C_2$ are semi-linearly rank-metric equivalent, and so the second result holds.

\end{proof}

\section{Automorphism Groups of Rank-Metric and Matrix Codes}\label{autgroupssection}

With the notion of rank-metric- and matrix-equivalence maps in place, we now examine the collection of such maps that leave a given code fixed.  This collection forms a group, known as the \emph{automorphism group} of the code, which is valuable for the enumeration of inequivalent codes.    
 \begin{definition}\label{autgroups}
%
 The \emph{linear rank-metric-automorphism group} of a rank-metric code $C \subseteq \F_{q^m}^l$ is the set of linear rank-metric-equivalence maps $f \in \LEquiv_{\RM}(\F_{q^m}^l)$ such that $f(C)=C$; this group is denoted $\LAut_{\RM}(C)$.  Similarly, the \emph{semi-linear rank-metric-automorphism group} of a rank-metric code $C \subseteq \F_{q^m}^l$ is the set of semi-linear rank-metric-equivalence maps $f \in \SLEquiv_{\RM}(\F_{q^m}^l)$ such that $f(C)=C$; this group is denoted $\SLAut_{\RM}(C)$.    
  
  The \emph{linear matrix-automorphism group} of a matrix code $C \subseteq \Matlm$ is the set of linear matrix-equivalence maps $f \in \LEquiv_{\Mat}(\Matlm)$ such that $f(C)=C$; this group is denoted $\LAut_{\Mat}(C)$.  Similarly, the \emph{semi-linear matrix-automorphism group} of a matrix code $C \subseteq \Matlm$ is the set of semi-linear matrix-equivalence maps $f \in \SLEquiv_{\Mat}(\Matlm)$ such that $f(C)=C$; this group is denoted $\SLAut_{\Mat}(C)$.    
  \end{definition}

\subsection{Rank-Metric Automorphism Groups of Gabidulin Codes}\label{RMAutGroup}
In this section, we will examine the linear rank-metric automorphism groups of Gabidulin codes when they are viewed as rank-metric codes; in the next section we will characterize the linear matrix-automorphism group of these codes when they are viewed as matrix codes via the map $\epsB$.  We focus on these codes specifically because they are the most well-known construction of rank-metric codes, and thus also the most widely used.  

Before developing the theory of automorphism groups of Gabidulin codes, we first review two results from  \cite{Berger03}. 
The first result characterizes precisely when two Gabidulin vectors $\vg$ and $\vg'$ determine the same Gabidulin code.
\begin{theorem}[\cite{Berger03}, Theorem 2]\label{BergerTheorem2}
 Let $\vg, \vg' \in \F_{q^m}^l$ be Gabidulin vectors.  For any $k$ with $1\leq k < l <m$, the Gabidulin codes $C_{k,\vg, q^m}, C_{k,\vg', q^m} \subseteq \F_{q^m}^l$ are equal if and only if there exists a scalar $\alpha \in \F_{q^m}^*$ such that $\vg'=\alpha \vg$.
\end{theorem}



The next lemma characterizes the effect of right multiplication by a non-singular matrix on a Gabidulin code, which will be useful in the sequel.
\begin{lemma}[Lemma 3 in  \cite{Berger03}]\label{BergerLemma3}
Let $f=[1, L] \in \LEquiv_{\RM}(\F_{q^m}^l)$, and fix $k$ with $1\leq k < l\leq m$.  For any Gabidulin vector $\vg\in \F_{q^m}^l$ with corresponding Gabidulin code $C_{k,\vg, q^m}$, we have $f(C_{k, \vg, q^m}) = C_{k,\vg L, q^m}$.
\end{lemma}

In Theorem \ref{LAutRMGab} we precisely characterize the linear rank-metric equivalence maps that fix a Gabidulin code, in other words, we give a complete characterization of the linear rank-metric automorphism group of a Gabidulin code.  Specifically, we show that the only linear rank-metric equivalence maps that fix a code are those of the form $[\alpha, M_\beta]$ where $\alpha \in \F_{q^m}^*$ and the matrix $M_\beta \in \GL_l(\F_q)$ is such that $\vg \cdot M_\beta = \beta \vg$, in other words the matrices that fix the code must produce the effect of scalar multiplication on the defining Gabidulin vector $\vg$.  

\begin{theorem}\label{LAutRMGab}
Let $1\leq k< l \leq m$.  Let $\vg  \in \F_{q^m}^l$ be a Gabidulin vector with corresponding Gabidulin code $C_{k, \vg, q^m}$.  Let $d$ be the largest integer such that $W:=\Span_{\F_q} \{g_1, \ldots, g_l\}$ is a vector space over $\F_{q^d} \subseteq \F_{q^m}$.  Then
\begin{enumerate}
\item $d$ divides $\gcd(l,m)$.
\item $\LAut_{\RM}(C_{k, \vg, q^m})~\cong~{ \scriptsize \left\{ \left[\alpha, M_{\beta}\right]~|~ \alpha \in \F_{q^m}^*,~\beta\in \F_{q^d}^*\right\}}$,\\
where $M_{\beta} = \left(\epsilon_{\vg}\left(\beta \vg\right)\right)^\top  := {\scriptsize \begin{bmatrix} \epsg(\beta g_1) \\ \vdots \\ \epsg(\beta g_l) \end{bmatrix}}^\top$.
\end{enumerate}
\end{theorem}

\begin{proof}
We will prove part 2 by double-containment; part 1 will fall out of the proof of 2 along the way.  

\noindent$(\subseteq)$ Let $f=[\alpha, L] \in \LAut_{\RM}(C_{k,\vg,q^m})$.  Using the fact that $C_{k,\vg, q^m}$ is $\F_{q^m}$-linear together with Lemma \ref{BergerLemma3}, we have:
\[
C_{k,\vg, q^m} = f\left(C_{k,\vg, q^m}\right) = \left( \alpha C_{k,\vg, q^m}\right) L = C_{k,\vg, q^m} L = C_{k,\vg L, q^m}.
\]
Thus, by Theorem \ref{BergerTheorem2}, there exists some $\beta \in \F_{q^m}^*$ such that $\vg L = \beta \vg$.  

Let $W$ be as in the statement of the theorem.  Since each entry of $\vg L = \beta \vg$ is a linear combination of $g_1, \ldots, g_l$, each entry of $\vg L=\beta \vg$ lies in $W$, i.e.\  $\beta g_i \in W$ for $1 \leq i \leq l$.   For any $w \in W$, there exist scalars $b_1, \ldots, b_l \subseteq \F_q$ such that $w= \sum_{i=1}^l b_i g_i$.  Observe $\beta w = \beta \sum_{i=1}^l b_i g_i = \sum_{i=1}^l b_i (\beta g_i)$, which is an element of $W$ since $\beta g_i \in W$ for $1 \leq i\leq l$, and so $W$ is closed under scalar multiplication by $\beta$.   Repeating this argument, we have that $W$ is closed under scalar multiplication by all positive powers of $\beta$.   Hence $W$ is a vector space over $\F_q(\beta)$.  

As in the statement of the theorem, let $\F_{q^d}$ be the largest subfield of $\F_{q^m}$ over which $W$ is a vector space.  Then 
\[
 l= \dim_{\F_q}W = (\dim_{\F_q} \F_{q^d}) (\dim_{\F_{q^d}} W) = d (\dim_{\F_{q^d}} W),
 \]
and so $d$ divides $l$.  But also since $\F_q \subseteq \F_{q^d} \subseteq \F_{q^m}$, we have that $d$ divides $m$.  Hence $d$ divides $\gcd(l,m)$, proving part 1 of the theorem.  Let $M_{\beta}$ be as in the statement.  Then $\vg M_{\beta} = \beta\vg$ since $\beta\vg = \epsg^{-1}(\epsg(\beta\vg)) = \vg \epsg(\beta \vg)^\top = \vg M_{\beta}$ by the definition of $\epsg$ and $\epsg^{-1}$.  Thus, since $\beta \vg= \vg L$, we have that $L = M_{\beta}$, and so $f$ has the desired form.

\noindent$(\supseteq)$  Let $f=[\alpha,M_{\beta}]$ for some $\alpha \in \F_{q^m}^*$ and $\beta \in \F_{q^d}^*$, where $M_{\beta}$ and $d$ are as in the statement of the theorem.  Since $f\left(C_{k,\vg, q^m}\right) = \left(\alpha C_{k,\vg, q^m}\right) M_{\beta} = C_{k,\vg, q^m}M_{\beta}$, we must show that $C_{k,\vg, q^m} = C_{k,\vg, q^m}M_{\beta}$.

Let $W$ be as in the statement of the theorem.  By hypothesis, $W$ is a vector space over $\F_{q^d}$, and so it is closed under multiplication by $\beta \in \F_{q^{d}}^*$. By the definition of $\epsilon_{\vg}^{-1}$, we have 
$\beta\vg = \epsg^{-1}(\epsg(\beta\vg)) = \vg \epsg(\beta \vg)^\top = \vg M_{\beta}$, and so $\vg M_{\beta} = \beta\vg$.
Thus, we have
\[
\begin{array}{lcll}
f(C_{k,\vg, q^m} )&=&\left( \alpha C_{k,\vg, q^m}\right)M_{\beta} &\\
&=&C_{k,\vg,q^m}M_{\beta} & \text{since } C_{k,\vg,q^m} \text{ is } \F_{q^m}\text{-linear}\\
 &=& C_{k, \vg  M_{\beta} , q^m}& \textrm{by Lemma }\ref{BergerLemma3}\\
&=& C_{k, \beta \vg, q^m}& \textrm{since } \vg M_{\beta}  = \beta \vg \\
&=& C_{k, \vg, q^m}& \textrm{by Theorem }\ref{BergerTheorem2}.
\end{array}
\]
Thus, $f(C_{k,\vg, q^m} )=C_{k,\vg, q^m} $, and so $f \in \LAut_{\RM}(C_{k,\vg,q^m})$.  
\end{proof}

This type of characterization of the linear rank-metric automorphism group of Gabidulin codes was previously attempted in \cite{Berger03}.   There, Berger claimed to show that the only rank-metric equivalence maps that fix a Gabidulin code have the form $f=[\alpha, L]$ where $\alpha \in \F_{q^m}^*$ and $L$ is a scalar matrix over $\F_q$ [\cite{Berger03}, Theorem 3].    However, we have found a flaw in his proof and as Theorem \ref{LAutRMGab} illustrates, there are significantly more matrices that will fix the code.  Specifically, there are a number of other matrices that accomplish scalar multiplication on the Gabidulin vector $\vg$ that defines the code beyond simply the scalar matrices.  The following example gives a sample of the additional types of matrices that are present in the linear automorphism group of a Gabidulin code, thereby illustrating the main result of Theorem \ref{LAutRMGab} and highlighting its differences from [\cite{Berger03}, Theorem 3].    

\begin{example}\label{autexample}
Write $\F_{16}=\F_2[\omega]$ where $\omega$ is a root of the primitive polynomial $p(t)=1+t+t^4$, and fix the ordered basis $\vb=(1, \omega, \omega^2,\omega^3)$ for $\F_{16}$ as an $\F_2$-vector space.
Let $C$ be the rank-metric code generated by the single vector $\vg=(1,~\omega^5 )$, i.e. 
\[
C =  \text{rowspan}_{\F_{16}} \vg = \left\{  \omega^i(1,~\omega^5 ) ~|~ 0\leq i \leq 14\right\} \cup \left\{\left(0,~0\right) \right\}
\]
Since the entries of $\vg$ are linearly independent over $\F_2$, we have that $C$ is a 1-dimensional Gabidulin code.  Since $\omega^5$ has order $3$ in $\F_{16}^*$, $\omega^5$ is a primitive element for the unique subfield of $\F_{16}$ that is isomorphic to $\F_4$.   Thus, $\text{span}_{\F_2}\{ 1, \omega^5 \}$ is a 1-dimensional vector space over $\F_4$, and so it is possible to write any $\F_4$-scalar multiple of the vector $\vg$ using linear combinations of $1$ and $\omega^5$; in other words for any $\beta \in \F_4^*$, there exists a matrix $M_{\beta} \in \GL_2(\F_2)$ such that $\beta \vg = \vg M_{\beta}$.  For example, consider $\beta = \omega^5$.  Then
\begin{eqnarray*}
M_{\beta} &=& {\scriptsize \begin{bmatrix} \epsg(\beta g_1) \\ \epsg(\beta g_2) \end{bmatrix}}^\top ={\scriptsize \begin{bmatrix} \epsg(\omega^5) \\ \epsg(\omega^{10}) \end{bmatrix}}^\top = {\scriptsize\begin{bmatrix} 0&1 \\ 1&1 \end{bmatrix}}
\end{eqnarray*}

Define $f=[1, M_{\beta}]$.  Since $1 \in \F_{16}^*$ and $M_{\beta} \in \GL_2(\F_2)$, $f \in \LEquiv_{\RM}(\F_{q^m}^l)$ by Proposition \ref{LEquivRM}.  Note that $M_{\beta}$ is not a scalar matrix, and so by Berger's previous assertion [\cite{Berger03}, Theorem 3], $f$ should not be an automorphism of $C$.  However, 
\begin{eqnarray*}
f(C)&=& \text{rowspan}_{\F_{16}} {\scriptsize\left((1,~\omega^5 ) \begin{bmatrix} 0&1 \\1&1 \end{bmatrix} \right)} \\
&=&\text{rowspan}_{\F_{16}}{\scriptsize \left(( \omega^5,~1+\omega^5)\right)}\\
&=&\text{rowspan}_{\F_{16}}{\scriptsize \left(( \omega^5,~\omega^{10}) \right)}\\
&=& \text{rowspan}_{\F_{16}}{\scriptsize\left( \omega^5 (1,~\omega^5 )\right)}\\
&=& C
\end{eqnarray*}
Thus, $f \in \LAut_{\RM}(C)$ even though $f$ is not of the form $[\alpha, \lambda I_l]$ for any $\alpha \in \F_{16}^*$ and $\lambda \in \F_2^*$.
\end{example}


\subsection{Matrix Automorphism Groups of Gabidulin Codes}\label{MatAutGroup}
This section gives a partial characterization of the matrix automorphism group of an expanded Gabidulin code.  The matrix automorphism group turns out to be much more complicated than the rank-metric automorphism group of a Gabidulin code because matrix equivalence is much more general than rank-metric equivalence, as was shown in Theorems \ref{LEquivRMLEquivMat} and \ref{SLEquivRMSLEquivMat}.  The following examples illustrate some of the key complications that arise in the characterization of the matrix automorphism group. 

Recall that a Gabidulin code $C_{k,\vg, q^m}$ is $\F_{q^m}$-linear, and so $\alpha C_{k,\vg, q^m} = C_{k,\vg, q^m}$ for any $\alpha \in \F_{q^m}^*$.   Furthermore, in Lemma \ref{BergerLemma3}, we saw that $C_{k,\vg, q^m} L = C_{k,\vg L, q^m}$ for any $L \in \GL_l(\F_q)$.  Thus, the image of $C_{k,\vg, q^m}$ under any linear rank-metric equivalence map $[\alpha, L]$ is another Gabidulin code. Example \ref{imagenotGab} shows that this property need not hold for linear matrix equivalence maps acting on expanded Gabidulin codes.

\begin{example}\label{imagenotGab} 
Write $\F_{64}=\F_2[\omega]$ where $\omega$ is a root of the primitive polynomial $p(t)=t^6+t^4+t^3+t+1$.     We will use a normal basis for $\F_{64}$ over $\F_2$ since such a basis interacts well with the structure of a Gabidulin code.  Although $\omega$ is not a normal element, $\omega^{38}$ is normal, and so we fix the ordered basis $\vb$ for $\F_{64}$ over $\F_2$ as
\[
\vb=\left(\omega^{38}, (\omega^{38})^2, (\omega^{38})^4, (\omega^{38})^8, (\omega^{38})^{16} , (\omega^{38})^{32}\right) = \left(\omega^{38}, \omega^{13}, \omega^{26}, \omega^{52}, \omega^{41}, \omega^{19}\right).
\]
Let $C_{2,\vg,64}$ be the 2-dimensional Gabidulin code generated by $\vg=\left( \omega^{37}, \omega^{42}, \omega^{16}, \omega \right)$, and set 
\[
{\scriptsize L = \begin{bmatrix}0&1& 1& 0\\0& 1& 0 &0\\0 &0& 0& 1\\ 1& 1& 1& 0 \end{bmatrix} \in \GL_4(\F_2)}, \hspace{.1in}
{\scriptsize M=\begin{bmatrix} 1& 0& 0& 0 &1& 0\\1& 1& 0& 1& 0 &1\\1 &1& 1& 1& 1& 1\\0& 1& 1& 0& 0& 0\\1& 1& 1& 0& 1& 1\\1& 0& 0& 1& 0& 0 \end{bmatrix} \in\GL_6(\F_q).}
\]
By Proposition \ref{LEquivMat}, $[L,M] \in \LEquiv_{\Mat}(\Matlm)$, and so $L \epsB(C_{2,\vg,64}) M$ is linearly matrix equivalent to $\epsB(C_{2,\vg,64})$.  When we examine $\hC:=\epsB^{-1}\left( L \epsB(C_{2,\vg,64}) M \right)$, we find that $\hC$ is not $\F_{64}$-linear because $|\Span_{\F_{64}}\{ \widehat{\vc} \in \hC\}| =16777216 > 4096 = |\hC|$.  Since every Gabidulin code over $\F_{64}$ is $\F_{64}$-linear, we see that $\hC$ cannot be a Gabidulin code.  Thus, $\epsB(C_{2,\vg,64})$ is linearly matrix equivalent to a code that is not the matrix expansion of a Gabidulin code.
\end{example}

While we cannot guarantee that every linear matrix equivalence will take an expanded Gabidulin code to another expanded Gabidulin code, we can guarantee that any linear matrix equivalence map that corresponds to a rank-metric equivalence map, namely any of the maps outlined in Theorem \ref{LEquivRMLEquivMat}, will have this property.  Thus, we obtain the following result.

\begin{proposition}\label{LAutMatSubgp}
Let $1\leq k< l \leq m$ and fix an ordered basis $\vb$ of $\F_{q^m}$ over $\F_q$.  Let $\vg  \in \F_{q^m}^l$ be a Gabidulin vector and let $d$ be as in Theorem \ref{LAutRMGab}.  Let $f \in \LEquiv_{\Mat}(\Matlm)$ be of the form $f(A) = LAM_{\alpha} \text{ for all }A \in \Matlm$ for some $L \in \GL_l(\F_q)$ and some $M_{\alpha}$ as in Lemma \ref{M_alphaQ}.  Then $f \in \LAut_{\Mat}(\epsB(C_{k,\vg, q^m}))$ if and only if $L= \epsilon_{\vg}\left(\beta \vg\right)$ for some $\beta\in \F_{q^d}^*$.  Hence, 
\[
\{(\epsilon_{\vg}\left(\beta \vg\right), M_{\alpha})~|~ \beta \in \F_{q^d}^*,~\alpha \in \F_{q^m}^*\} \subseteq \LAut_{\Mat}(\epsB(C_{k,\vg, q^m})).
\]
\end{proposition}
\begin{proof}
By Theorem \ref{LEquivRMLEquivMat}, $\epsB^{-1}(f(\epsB(C_{k,\vg, q^m})))$ is rank-metric equivalent to $C_{k,\vg, q^m}$; specifically, by Lemmas \ref{M_alphaQ} and \ref{epsBL}, 
\[
\epsB^{-1}(f(\epsB(C_{k,\vg, q^m}))) = \epsB^{-1}(L \epsB(C_{k,\vg, q^m}) M_\alpha) = \alpha C_{k,\vg, q^m} L^\top = g(C_{k,\vg, q^m})
\]
where $g=(\alpha, L^\top) \in \LEquiv_{\RM}(\F_{q^m}^l).$  Thus, since $f \in \LAut_{\Mat}(\epsB(C_{k,\vg, q^m}))$ if and only if $f(\epsB(C_{k,\vg, q^m})) = \epsB(C_{k,\vg, q^m})$ if and only if $\epsB^{-1}(f(\epsB(C_{k,\vg, q^m}))) = C_{k,\vg, q^m}$ if and only if $g(C_{k,\vg, q^m}) = C_{k,\vg, q^m}$, we see
$f \in \LAut_{\Mat}(\epsB(C_{k,\vg, q^m}))$ if and only if $g \in \LAut_{\RM}(C_{k,\vg, q^m})$.  By Theorem \ref{LAutRMGab}, $g \in \LAut_{\RM}(C_{k,\vg, q^m})$ if and only if $g$ has the form $ \left(\gamma, \left(\epsilon_{\vg}\left(\beta \vg\right)\right)^\top\right)$ for some $\gamma \in \F_{q^m}^*$ and some $\beta \in \F_{q^d}^*$.  Thus, $g \in \LAut_{\RM}(C_{k,\vg, q^m})$ if and only if $L^\top =  \left(\epsilon_{\vg}\left(\beta \vg\right)\right)^\top$, and so the result holds. 
\end{proof}

A key feature of the rank-metric automorphism group of a Gabidulin code that we saw in Theorem \ref{LAutRMGab} is that the automorphism group has a direct product structure modulo the subgroup $N=\{(\lambda, \lambda^{-1}I_n)~|~\lambda~\in~\F_q^*\}$.  In other words, $[\alpha, L] \in \LAut_{\RM}(C_{k,\vg, q^m} )$ if and only if  $[\alpha, I_l] \in \LAut_{\RM}(C_{k,\vg, q^m} )$ and $[1, L] \in \LAut_{\RM}(C_{k,\vg, q^m} )$.  Proposition \ref{LAutMatSubgp} shows that $ \LAut_{\Mat}(\epsB(C_{k,\vg, q^m} ))$ contains a subgroup with this same direct product structure, but this does not guarantee that $ \LAut_{\Mat}(\epsB(C_{k,\vg, q^m} ))$ as a whole has this direct product structure.  In fact, Example \ref{notdirectproduct} gives an explicit example of a matrix automorphism group of an expanded Gabidulin code that violates this direct product structure.

\begin{example}\label{notdirectproduct}
We consider the same Gabidulin code from Example \ref{imagenotGab}, expanded with respect to the same basis for $\F_{64}$ over $\F_2$.  Set
\[
{\scriptsize L = \begin{bmatrix}0&1& 1& 0\\0& 1& 0 &0\\0 &0& 0& 1\\ 1& 1& 0& 0 \end{bmatrix} \in \GL_4(\F_2)}, \hspace{.4in}
{\scriptsize M=\begin{bmatrix} 0 &1 &0& 1& 0& 1\\0& 1& 0& 0& 1& 0\\0& 1& 0& 1& 0& 0\\1& 1& 1& 1 &1& 1\\0& 1& 0& 0& 0& 0\\ 1& 1 &0 &1 &1 &0 \end{bmatrix} \in\GL_6(\F_q).}
\]
One can check that $[L,M] \in \LAut_{\Mat}(\epsB(C_{2,\vg,64}))$.  When we examine the effect of $L$ alone, and apply Lemma \ref{BergerLemma3}, we find
\[
\begin{array}{lll}
\epsB^{-1}(L \epsB(\vg)) &=&\vg L^\top\\
&  =& (\omega, \omega^{37}+\omega^{42}+\omega, \omega^{37}, \omega^{16})\\
& =& (\omega, \omega^{14}, \omega^{37}, \omega^{16}) \notin C_{2,\vg,64},
\end{array}
\]
and so $L \epsB(C_{2,\vg,64}) \neq \epsB(C_{2,\vg,64})$.  Thus, $[L, I_m] \notin \LAut_{\Mat}(\epsB(C_{2,\vg,64}))$, and so $\LAut_{\Mat}(\epsB(C_{2,\vg,64}))$ does not have the structure of a direct product.
\end{example}

Thus, the group structure of $ \LAut_{\Mat}(\epsB(C_{k,\vg, q^m} ))$ is significantly more complicated than that of $ \LAut_{\RM}(C_{k,\vg, q^m} )$, which seems to be a reflection of the fact that matrix equivalence is strictly more general than rank-metric equivalence, as seen in Theorem \ref{LEquivRMLEquivMat}.

\section{Discussion}
Given the growing number of applications for random linear network coding, it is essential that methods of providing error correction for this form of network coding be further investigated.  K\"{o}tter and Kschischang demonstrate the error-correcting value of subspace codes; in particular, they establish the near-optimality of lifted rank-metric and lifted matrix codes in this context \cite{KK08}.  Since  lifted rank-metric and lifted matrix codes inherit their structure and distance distributions from the underlying rank-metric and matrix codes, further examination of these aspects of the underlying rank-metric and matrix codes is essential.  Toward this end, this paper has created a framework for classifying rank-metric and matrix codes in terms of their structural and distance properties.  This was accomplished by defining a notion of equivalence that preserves these properties and characterizing the sets of linear and semi-linear equivalence maps for both rank-metric- and matrix-equivalence.  We also characterize the subset of linear rank-metric equivalence maps that fix the family of rank-metric codes known as Gabidulin codes, and provide a partial characterization of the linear matrix-equivalence maps that fix the matrix codes obtained by expanding Gabidulin codes with respect to an arbitrary basis for the extension field over which these codes are defined.  One area of future research is to provide a complete characterization of this linear matrix-automorphism group for the expanded Gabidulin codes.  

Public-key cryptography provides another venue in which this analysis of the linear equivalence maps and their action on Gabidulin codes may prove valuable.  As outlined in Section \ref{background}, Gabidulin codes are widely used for generating subspace codes, but they have also found applications in defining a public-key cryptosystem, known as the GPT cryptosystem, analogous to the McEliece cryptosystem \cite{Gab08}.  In this setting, Gabidulin codes have proven valuable because they have high minimum distance and an efficient decoding algorithm, but are resistant to combinatorial decoding attacks by cryptanalysts when the code in use is unknown.  One drawback of these codes, however, is that their highly structured nature enables cryptanalysts, via Overbeck's attack \cite{Overbeck}, to recover the original code and crack the cryptosystem.  To attempt to disguise the structure of the code, a simple rank-metric-equivalence map, namely a permutation matrix over the base field, is employed in one updated version of the GPT cryptosystem; however, the permutation matrix still does not provide sufficient protection to resist Overbeck's attack \cite{Gab08}.  To circumvent this attack, Gabidulin proposed using a permutation matrix over an extension field, which no longer guarantees that the modified code will be equivalent to the original Gabidulin code, and thus the high minimum-distance property may be lost.  One possible alternative to this is the use of matrix-equivalence maps acting on the expanded Gabidulin code as a means to further disguise the structure of Gabidulin codes while still maintaining the distance distribution.  This is an important area of future research, in particular because the GPT cryptosystem has the potential to be a public-key cryptosystem that is impervious to the advent of quantum computing unlike the now-commonly used RSA public-key cryptosystem.  

\section{Acknowledgments}
The results in this paper are taken from the author's thesis \cite{MyThesis}.  The author would also like to thank her advisor Judy L.\ Walker for her continued help and insight.
 \bibliographystyle{plain}
 \bibliography{IEEE_Bibliography}
%
%
%
\end{document}